\documentclass[a4paper,UKenglish,autoref,thm-restate]{lipics-v2021}

\nolinenumbers

\title{Reducing Quantum Circuit Synthesis to \#SAT}

\author{Dekel Zak}{Leiden University, The Netherlands}{d.z.zak@umail.leidenuniv.nl}{https://orcid.org/0009-0001-0315-1539}{}
\author{Jingyi Mei}{Leiden University, The Netherlands}{j.mei@liacs.leidenuniv.nl}{https://orcid.org/0000-0002-4665-9818}{}
\author{Jean-Marie Lagniez}{Artois University, France}{lagniez@cril.fr}{https://orcid.org/0000-0002-6557-4115}{}
\author{Alfons Laarman}{Leiden University, The Netherlands}{a.w.laarman@liacs.leidenuniv.nl}{https://orcid.org/0000-0002-2433-4174}{}

\authorrunning{D. Zak, J. Mei, J.M. Lagniez \&  A. Laarman} 

\Copyright{Dekel Zak, Jingyi Mei, Jean-Marie Lagniez \& Alfons Laarman} 



\category{} 

\relatedversion{} 

\supplement{The source code to the tool hosted on a GitHub repository}\supplementdetails[
subcategory={Software, Source Code},
]{Software}{https://github.com/System-Verification-Lab/Quokka-Sharp}

\funding{This work is supported by the Dutch National Growth Fund, as part of the Quantum Delta NL program.}

\EventEditors{Maria Garcia de la Banda}
\EventNoEds{1}
\EventLongTitle{31st International Conference on Principles and Practice of Constraint Programming (CP 2025)}
\EventShortTitle{CP 2025}
\EventAcronym{CP}
\EventYear{2025}
\EventDate{August 10--15, 2025}
\EventLocation{Glasgow, Scotland}
\EventLogo{}
\SeriesVolume{340}
\ArticleNo{39}

\RequirePackage{color}

\RequirePackage{physics}
\RequirePackage{graphicx}
\RequirePackage{amsmath}

\RequirePackage{braket}

\RequirePackage{amsthm}
\RequirePackage{thmtools}
\RequirePackage{xspace}
\RequirePackage{multirow}
\usepackage{colortbl}
\RequirePackage{stmaryrd}
\usepackage{appendix}
\RequirePackage{calc}
\RequirePackage{nicefrac}
\RequirePackage{complexity}
\RequirePackage{wrapfig}

\usepackage{adjustbox}
\usepackage{array}
\usepackage{booktabs}
\usepackage{multirow}
\usepackage{makecell}
\usepackage{empheq}

\usepackage{tabularx}
\usepackage{etoolbox}
\usepackage{arydshln}
\usepackage{mathtools}
\usepackage{thmtools}
\usepackage{xparse}
\let\paragraph\relax
\usepackage{etoolbox}

\usepackage{pifont}
\RequirePackage[bottom,symbol]{footmisc} 
\RequirePackage{footnote}

\usepackage{bm}

\renewcommand{\dagger}{\ensuremath{\text{\textdagger}}} 

\newcommand\defmath[2]{\newcommand#1{\ensuremath{#2}\xspace}}
\newcommand\concept[1]{\textit{#1}}

\newcommand\smat[1]{\ensuremath{\begin{smallmat} #1 \end{smallmat}}\xspace}

\newcommand\mat[1]{\ensuremath{\begin{bmatrix*}[r]#1\end{bmatrix*}}\xspace}

\renewcommand\dots{\makebox[.7em][c]{.\hfil.\hfil.}}

\newtheorem{problem}[theorem]{Problem}

\renewcommand\phi{\varphi}

\let\set\undefined

\providecommand{\set}[1]{\ensuremath{\left\lbrace #1 \right\rbrace}}
\providecommand{\sequence}[1]{\ensuremath{\left( #1 \right)}}

\providecommand{\abs}[1]{\lvert#1\rvert}

\newcommand{\defn}{\,\triangleq\,}

\newenvironment{smallmat}{\left[\begin{smallmatrix}}{\end{smallmatrix}\right]}

\newcommand\hide[1]{}

\defmath\samp{\textbf{Sample}}
\defmath\pro{\textbf{Measure}}
\defmath\eq{\textbf{Equal}}
\defmath\res{\textbf{Res}}
\defmath\addi{\textbf{Addition}}
\defmath\inprod{\textbf{InnerProd}}
\defmath\fid{\textbf{Fidelity}}
\defmath\had{\textbf{Hadamard}}
\defmath\xyz{\textbf{X,Y,Z}}
\defmath\cx{\textbf{CX}}
\defmath\cz{\textbf{CZ}}
\defmath\loc{\textbf{Local}}
\defmath\T{\textbf{T}}

\makeatletter
\DeclareRobustCommand{\qed}{\ifmmode
    \eqno \def\@badmath{$$}\let\eqno\relax \let\leqno\relax \let\veqno\relax
    \hbox{\openbox}\else
    \leavevmode\unskip\penalty9999 \hbox{}\nobreak\hfill
    \quad\hbox{\openbox}\fi
}
\makeatother

\newcommand{\circuit}{\mathcal{C}}
\newcommand{\gateset}{\mathcal{G}}
\newcommand{\jfid}{\text{Jamio\l{}kowski fidelity}}
\newcommand{\pauli}{\mathcal{P}}
\defmath\Rot{\mathit{Rot}}

\defmath\init{\mathit{init}}
\defmath{\qg}{\mathbf{G}}
\newcommand{\hh}{\mathcal{H}}
\newcommand{\Fid}{\mathrm{Fid}}
\defmath{\combasis}{\mathbf{CB}}
\defmath{\paulibasis}{\mathbf{PB}}

\defmath{\toffoli}{\mathit{Toffoli}}

\defmath{\mitms}{\textbf{mitms}}

\defmath\q{{\vec q}}
\defmath\aux{{\vec u}}

\defmath\sat{SAT}
\defmath\mc{\#SAT}
\defmath\wmc{\#SAT_W}
\defmath\WeightedMax{\mathit{WeightedMax}}
\defmath\Max{\mathit{Max}}
\defmath\maxmc{\Max\mc}
\defmath\maxwmc{\Max\wmc}

\newcommand{\aref}[1]{\hyperref[#1]{Appendix~\ref*{#1}}}

\defmath\true{1}
\defmath\false{0}

\newcommand\no[1]{\ensuremath{\overline{#1}}}

\begin{document}

\renewcommand\subsectionautorefname{Section}
\renewcommand\subsubsectionautorefname{Section}

\maketitle

\begin{abstract}
    Quantum circuit synthesis is the task of decomposing a given quantum operator into a sequence of elementary quantum gates. 
    Since the finite target gate set cannot exactly implement any given operator, approximation is often necessary.
    Model counting, or \#SAT, has recently been demonstrated as a promising new approach for tackling core problems in quantum circuit analysis.
    In this work, we show for the first time that the universal quantum circuit synthesis problem can be reduced to maximum model counting.
    We formulate a \#SAT encoding for exact and approximate depth-optimal quantum circuit synthesis into the Clifford+T gate set.
    We evaluate our method with an open-source implementation that uses the maximum model counter d4Max as a backend. For this purpose, we extended d4Max with support for complex and negative weights to represent amplitudes.
    Experimental results show that existing classical tools have potential for the quantum circuit synthesis problem.
\end{abstract}

\section{Introduction}

Quantum algorithms are typically specified using higher-level constructs such as (classical) procedures and quantum Fourier transform (QFT)~\cite{kliuchnikov2016practical}.
For their efficient implementation on physical devices, which typically operate with a finite gate set, they need to be broken down into a quantum circuit, a sequence of quantum gates.
Error-corrected hardware implementations, which realize the ideal quantum computational model formalized in a quantum Turing machine~\cite{bernstein1993quantum}, 
often use the universal Clifford+T set $\{S, H, CX, T\}$~\cite{fowler2012surface,aliferis2006quantum,nielsen2000quantum}.
Moreover, better optimal synthesis solutions can have significant implications for quantum complexity theory, such as reducing the stabilizer rank of magic gadgets~\cite{kocia2020improved}, and relating classical and quantum resources~\cite{bravyi2019simulation,bravyi2016trading}.
However, exact synthesis is not always feasible due to the discrete nature of this gate set, often necessitating approximation techniques to balance precision and circuit depth~\cite{kliuchnikov2016practical}. The computational complexity of synthesis is formidable, with classical approaches exhibiting doubly exponential time in the worst case~\cite{kuncak2010complete}, underscoring the need for scalable and efficient methods.

Classical methods, like decision diagram~\cite{limdd,burgholzer2020advanced,sistla2023symbolic}, tree-automata~\cite{chen2023automata}, and SAT~\cite{schneider2023sat} have proven to be highly effective in analyzing quantum circuits.
Recent advances have highlighted weighted model counting, or \#SAT, as a powerful tool for addressing hard problems in quantum circuit analysis, including simulation~\cite{mei2024simulating} and equivalence checking~\cite{mei2024eq}. These methods leverage off-the-shelf solvers to tackle the exact versions of these problems, known as \#\P-complete~\cite{nest2008classical}.
Building on this promise, we explore whether \#SAT can be harnessed for the universal quantum circuit synthesis problem ---a challenge that also requires approximation since discrete universal gate sets cannot exactly implement arbitrary quantum operators. The relevant complexity class here is \QMA{}, as approximate circuit equivalence checking is complete for it~\cite{ji2009non}. 
Moreover, the Solovay-Kitaev theorem~\cite{dawson2005solovaykitaev} guarantees that any unitary can be $\epsilon$-approximated with a $\polylog(\frac1\epsilon)$-depth circuit.

This work presents a novel reduction of quantum circuit synthesis to maximum model counting, focusing on depth-optimal and approximate synthesis into the Clifford+T gate set.
To achieve this, we overcome two obstacles.
First, to enable reduction to the maximum weighted model counting, we show how equivalence checking can be performed with a single model counting call, contrasting previous approaches~\cite{mei2024eq} which required linearly many calls. In the process, we provide two new ways of encoding equivalence checking. Second, we show how the fidelity (a measure of circuit similarity) between two circuits can be computed by weighted model counting, generalizing our exact equivalence checking methods to support approximate equivalence checking.
Our approach then encodes the synthesis task as a weighted conjunctive normal form (CNF) formula, where satisfying assignments correspond to valid gate sequences and weights reflect approximation fidelity. We demonstrate that this reduction enables both exact and approximate synthesis, with applications to circuit optimization. We provide an open-source implementation, called \texttt{Quokka\#-syn}, to validate our method.

The scalability of quantum circuit synthesis remains a critical bottleneck, with existing methods struggling to handle large qubit counts or gate depths without resorting to corner-case omissions~\cite{matteo2016parallel,paradis2024synthetiq}.
We experimentally evaluate our method to compare the different encodings and test its scalability. In addition, we compare it against a state-of-the-art tool focusing on depth-optimal approximate and exact synthesis.
While the comparison has limitations, it shows that our \#SAT-based approach has merit and can offer improvement over existing methods. However, it falls short of fully resolving the synthesis problem’s inherent complexity. We identify multiple avenues for improvement, such as enhancing maximum model counters with support for incremental solving. By laying this groundwork, we pave the way for integrating advanced classical solvers into the quantum computing toolkit, advancing the practical realization of quantum algorithms.

 \section{Preliminaries}
\subsection{Quantum Computing}\label{sec:pre-quan}
In this section, we give the necessary notions and notations on quantum computing.
For an entailed explanation, for example on tensor product ($\otimes$), see \cite{nielsen2000quantum}.

\textbf{Quantum states.}
Let $\mathcal{H}^{\otimes n}$ be the $2^n$-dimensional Hilbert space.
An $n$-qubit quantum state, denoted as $\ket{\phi}$ using Dirac notation, is a column vector 
$\mat{\alpha_{00\dots00},\dots,\alpha_{11\dots11}}^T$ in $\mathcal{H}^{\otimes n}$,
where $|\alpha_b|^2 \in [0,1]$ are \textit{amplitudes}, satisfying: 
$|\alpha_{00\dots00}|^2 +\dots + |\alpha_{11\dots11}|^2 = 1$.
Its complex adjoint $\bra{\phi}$ is a row vector with conjugated entries:
$\bra{\varphi} = \ket{\varphi}^\dagger = $\mat{\alpha_{00\dots00}^*,\dots,\alpha_{11\dots11}^*}, therefore $\braket{\phi|\phi} = 1$.
A quantum state vector can be decomposed in the \emph{computational basis},
written as 
$
\ket{\varphi} = \sum_{b\in\{0,1\}^n}\alpha_b \ket{b}
$,
where $\ket{b}$ is a \textit{computational basis state} defined as a vector with all entries setting to~0 except at index $b$ setting to 1.

Another way to represent a state $\ket\phi$ is as a \emph{density matrix} $\rho =\ketbra{\phi}$.
The \emph{trace} of a density matrix is $1$, denoted as $\tr(\rho) = 1$,
where trace is defined as the sum of the diagonal elements of a matrix.
A density matrix can be decomposed in the \emph{Pauli basis} as in \autoref{eq:pauli}.
To understand Pauli basis decomposition,
we introduce \emph{Pauli matrices} and \emph{Pauli strings}.
The Pauli matrices are $I = \smat{1&0\\0&1}, X = \smat{0&1\\1&0}, Y= \smat{0&-i\\i&0}, Z=\smat{1&0\\0&-1}$. 
An $n$-qubit \concept{Pauli string} $\mathcal{P}$ is a parallel composition of $n$
Pauli gates, such that $\mathcal{P} \in \set{I, X, Y, Z}^{\otimes n}$.
For instance, $X\otimes Z\otimes I$ is a three-qubit Pauli string.
It is worth noting that while a density matrix $\rho$ may contain complex numbers,
the basis \textit{coefficients} $\beta_i$ in \autoref{eq:pauli} are all real numbers~\cite{mei2024eq}.
Throughout this paper, we denote $[n] = \{0,\dots, n-1\}$.
\begin{equation}\label{eq:pauli}
\rho = \sum_{j\in [4^n]} \beta_j \cdot \mathcal{P}_j
\text{~~~ for the Pauli strings ~~}
\mathcal{P}_j \in \set{I, X, Y, Z}^{\otimes n}
\end{equation}
\textbf{Quantum gates.}
An $n$-qubit quantum gate $G$ can be expressed by a $2^n\times 2^n$ unitary matrix $U$, i.e., $U^\dagger \cdot U = U \cdot U^\dagger = I^{\otimes n}$. 
A single qubit quantum gate $U$ operating on qubit $j\in[n]$ can be represented as $U_j = I^{\otimes j} \otimes U \otimes I^{\otimes n - j -1}$.
Updating a quantum state in vector form $\ket{\varphi}$ is done by matrix-vector multiplication, i.e. $\ket{\psi}=U\ket{\varphi}$.
Applying a unitary $U$ to a density matrix $\rho =\ketbra{\phi}$ should be done through \concept{conjugation}, i.e.: $U\rho U^\dagger = 
U\ketbra{\phi} U^\dagger = \ketbra{\psi}$ for $\ket\psi = U \ket\phi$.
We consider the well-known universal gate set Clifford \set{S,H,\textit{CX}} + $T$ and the gates \toffoli and $CH$,
which are defined as follows:
\[
H = \nicefrac{1}{\sqrt{2}} \begin{smat} {1&1\\1&-1}\end{smat}, \
S = \begin{smat}{1& 0\\0&i}\end{smat}, \
CX = \begin{smat}{I& \mathbf{0}\\ \mathbf{0}&X}\end{smat}, \
T = \begin{smat}{1& 0\\0&i}\end{smat},
\text{$\toffoli$} = \begin{smat}{I_6 & \mathbf{0}\\ \mathbf{0}&X}\end{smat}, \
CH = \begin{smat}{I& \mathbf{0}\\ \mathbf{0}&H}\end{smat}
\]
where $I_k$ is a $k$-dimensional identity matrix and $\mathbf{0}$ is the all zero matrix. We ignore the index when the dimension is 2, i.e. $I = I_2= \begin{smat}{1& 0\\0&1}\end{smat}$.
We denote a gate set as $\gateset$.
For an $n$-qubit circuit,
let $\gateset(k)\subseteq \gateset$ be a subset of gates applied to $k\in[n]$ qubits,
such that $\gateset(1)$ is the set of single-qubit gates, $\gateset(2)$ is the set of two-qubit gates and so on.
For example, if $\gateset = \{H,T,CX, \toffoli\}$,
$\gateset(1) = \{H,T\}$, $\gateset(2) = \{CX\}$ and $\gateset(3) = \{\toffoli\}$. We always assume $I\in\gateset$, also if not explicitly stated. 

\textbf{Quantum circuits.} The evolution of a quantum system is modeled by a \emph{quantum circuit}, a sequence of \textit{quantum gate layers}, or \emph{layers} in short, applied to all qubits at each time step. A \textit{layer} $D$ is a set of gates such that each qubit has at most one gate applied to it.
Thus, it contains only mutually parallel gates. 
For example, a \emph{layer} $D = \{G_i \mid i\in[n]\}$ applies a single-qubit local gate $G$ on each qubit $i$.
For any layer $D$, we will denote its unitary matrix with $U_{{D}}$.
Let $\mathcal D_n$ be the set of all possible layers for $n$ qubits.
For a \emph{circuit} $\mathcal{C} = (D^0, \dots, D^{d-1})$, with $D^i \in \mathcal D_n$, we thus have its unitary $U_{\mathcal{C}} = U_{D^{d-1}}\cdots U_{D^0}$.
We define a circuit's \emph{depth} as the minimal number of layers (with parallel gates) that it contains.

\paragraph{\jfid{}} 
The \emph{fidelity} between two quantum states $\ket{\psi}$ and $\ket{\varphi}$ is defined as
$\Fid(\ket{\varphi},\ket{\psi})\defn
\tr(\ketbra{\psi}\cdot\ketbra{\phi})=|\braket{\psi|\varphi}|^2$.
The fidelity between states can be extended to measure the distance between unitary operators with the help of \emph{Jamio\l{}kowski isomorphism} that maps unitary $U$ on $\hh^{\otimes n}$ to a state $\ket{\varphi_U} = (U\otimes I)\ket{\Psi_n}$ on $\hh^{\otimes{2n}}$,
where $\ket{\Psi_n} = \frac{1}{\sqrt{2^n}}\sum_{i\in\{0,1\}^n}\ket{ii}$ and $\ket{ii}$ is short for $\ket{i}\otimes\ket{i}$.
Thus, \emph{\jfid{}}~\cite{raginsky2001fidelity} between unitary operators can be formally defined as:
\begin{equation}\label{Jfid}
    \Fid_J(U,V) = \Fid(\ket{\varphi_U}, \ket{\varphi_V})=|\braket{\varphi_U|\varphi_V}|^2.
\end{equation}
In the particular case where $\Fid_J(U,V) = |\braket{\varphi_U|\varphi_V}|^2 = 1$, it follows that $U = \lambda V$ with $|\lambda|^2 = 1$. This result arises because $|\braket{\varphi_U|\varphi_V}|^2 = 1$ implies $\ket{\varphi_U} = \lambda \ket{\varphi_V}$, and from the definitions of $\ket{\varphi_U}$ and $\ket{\varphi_V}$, we conclude $U = \lambda V$.

\subsection{Maximum Weighted Model Counting}

Let $A$ be a set of Boolean variables $\{a_1, \dots, a_m\}$.
A \emph{literal} $\ell$ is a variable or its negation, e.g., $a$ or $\neg a$ (written as $\overline a$).
A \emph{clause} is a disjunction of literals $\ell_1 \lor \dots \lor \ell_m$.
A propositional formula $F$ in \emph{Conjunctive Normal Form} (CNF) is a conjunction of clauses.
We write $F(A)$ to indicate that $F$ is defined on the set of variables $A$.
Let $\mathbb{B}$ be $\{0,1\}$ and $\mathbb{R}$ be the set of real numbers.
An \emph{assignment} $\tau$ maps variables $A$ to $\mathbb{B}$, where $0,1$ represents $\mathit{False}$ and $\mathit{True}$ respectively.
The satisfiability problem is determining whether an assignment $\tau$ to $A$ exists, for which a propositional formula $F(A)$ is true.
We denote $\sat(F(A)) \defn \{ \tau \mid \tau \text{ satisfies }F(A) \}$ as the
set of all satisfying assignments for $F$. The model counting problem is to compute the number of satisfying assignments,
denoted as $\#SAT(F) \defn |SAT(F)|$.

A weight function $W: A\times \mathbb{B} \rightarrow \mathbb{R}$ maps a variable with its true or false assignment, or, viewed alternatively, a literal, to a weight.
Given an assignment $\tau$, the weight of this assignment,
written $W(\tau)$, is the product of the weight of each variable with its
assignment $W(\tau) = \prod_{a\in A} W(a,\tau(a))$.
For notational convenience, we write the weights of literals, such that for a variable $a\in A$, we denote $W(a) = W(a, 1)$ and $W(\bar{a}) = W(a, 0)$.
If unspecified, the weight of a variable $a\in A$ is assumed to be $W(a) = W(\bar{a}) = 1$. We call such variables \concept{unbiased}.
The weighted model counting (WMC) problem asks for the sum of the weights of the satisfying assignments,
i.e., $\#SAT_W(F) = \sum_{\tau\in {SAT}(F(A))} W(\tau)$.

The maximum weighted model counting problem (MWMC) extends WMC by finding an assignment to a subset of variables that maximizes the weight of the WMC problem (\autoref{def:mwmc}).
In \autoref{sec:exactsyn}, we show that quantum circuit synthesis can be reduced to it. 

\begin{definition}[MWMC \cite{audemard2022maxsat}]\label{def:mwmc}
Given a propositional formula $F(A, B)$
over disjoint sets of variables $A$ and $B$, 
and a weight function W over ${(A \cup B)}\times \mathbb{B}$,
the MWMC problem is to determine an assignment $\tau$ to $A$ that maximizes $\wmc(F(A,B))$.
\end{definition}
For notation, we define an oracle function \maxmc,
which takes the quantified Boolean formula $F(A, B)$ with its weight function $W$ and returns an assignment $\tau$ to $A$ that maximizes the weighted model count of the formula $F$ and its maximal weight $w_{\textit{max}}$:
\[\maxwmc(F(A, B)) = (\tau(A), w_{max}).\]

\subsection{Reducing Quantum Computing to Model Counting}\label{sec:pre-encode}

We present two encodings: a computational basis encoding, referred to as \combasis, where the computational basis decomposition of the state vector is encoded directly, and a Pauli basis encoding, referred to as \paulibasis, where the basis states are the Pauli strings and we encode the density matrix decomposition (see \autoref{eq:pauli}).
A quantum state is encoded as a set of satisfying assignments where each satisfying assignment with its weight represents a \emph{basis state} with its \emph{amplitude} (in \combasis) or \emph{coefficient} (in \paulibasis).
Here, we briefly introduce both.
For details, we refer to~\cite{mei2024simulating,mei2024eq,mei2024gapp}.

For an $n$-qubit quantum state $\ket{\phi}$, we denote its WMC encoding with $F_{\ket{\phi}}$.
Note that in the \texttt{\paulibasis} basis, quantum states are represented using density matrices, 
so the corresponding encoding is technically \( F_{\ketbra{\phi}} \). 
However, we abuse the notation and write \( F_{\ket{\phi}} \) in both bases for simplicity.
We reserve propositional variables
$\q = \sequence{q_0, \dots, q_{n-1}}$ in the \combasis,
and $\q = \sequence{\q_0, \dots, \q_{n-1}} = \sequence{x_0, z_0, \dots, x_{n-1}, z_{n-1}}$ in the \paulibasis.
(Here, we let $\overline{x z} \equiv I$, $x\overline{z} \equiv X$, ${x z} \equiv Y$ and $\overline{x} z \equiv Z$, and $\overline{x_k} z_k\land \bigwedge_{i\in [n]\setminus \set{k}} \overline{x_i z_i}  \equiv Z_k$, etc, as in~\cite{aaronson2008improved}.)
The variables in $\vec q$ remain unbiased.
In addition, when needed, we introduce auxiliary variables to represent weights.
Since the assignment to these auxiliary variables is always determined by the assignment to \q, we often omit these variables, writing $F_{\ket{\phi}}(\q)$ instead of $F_{\ket{\phi}}(\q, \aux)$.
\autoref{tab:state-encoding} gives encoding examples.

\setlength{\tabcolsep}{3pt}
\renewcommand{\arraystretch}{1.4}
\begin{table}[b!]
    \caption{State encoding in both bases: The auxiliary variables, marked in gray, depend fully on the unbiased variables $\vec q$ (in \combasis) and $\vec x, \vec z$ (in \paulibasis) representing basis states.}
    \label{tab:state-encoding}\vspace{-1em}
\footnotesize
    \centering
    \begin{tabular}{c|cc|c}
    \toprule
        & \bf Comp. Basis (\combasis) & \bf Pauli Basis (PB) & \bf Auxiliary weights \\
    \midrule
        Variables & $\q = \sequence{q_0, \dots, q_{n-1}}$ &  $\q = \sequence{x_0, z_0, \dots, x_{n-1}, z_{n-1}}$  &  $h, r, g, w$  \\
\midrule
        $\ket{0}  \equiv$
            & $F_{\ket{0}}(q) = \no{q} $
            & $F_{\ket{0}}(x, z\textcolor{gray}{,g}) = F_{\frac{I+Z}{2}} =  \no{x}g$   & $W(g) =(\nicefrac1{ 2})$ \\
        $\ket{-}  \equiv$
            & $F_{\ket{-}}(q \textcolor{gray}{, h, r}) = h  {(r\Leftrightarrow q)} $ 
            & $F_{\ket{-}}(x, z\textcolor{gray}{,g}) = F_{\frac{I-X}{2}} = \no{z}g$ 
            & $W(h) =(\nicefrac1{\sqrt 2})$, $W(r) = (-1)$ \\
        \cite{huang2019approximate}: $\ket{\A}  \equiv$
            &  $F_{\ket{A}}(q \textcolor{gray}{, h, w}) = h (w \Leftrightarrow q)$
            & $ F_{\ket{A}}(x, z\textcolor{gray}{,g, h}) = 
            F_{\frac I2 + \frac{X+Y}{\sqrt 2}}  $
&    $W(w)= \omega$ where $\omega = \sqrt i = \frac{i+1}{\sqrt 2} $\\
        $\ket{00}  \equiv$ & $F_{\ket{00}}(q_0, q_1) = \no{q_0} \no{q_1} $  & $F_{\ket{00}}(x_0, z_0,x_1,z_1) = \no{x_0} \no{x_1}$  & none \\
        $\ket{\text +\text +} \equiv$ & $F_{\ket{\text +\text +}}(q_0, q_1\textcolor{gray}{,g}) = g $  & $F_{\ket{\text +\text +}}(x_0, z_0,x_1,z_1) = \no{z_0} \no{z_1}$ & $W(g) =(\nicefrac1{ 2})$ \\
    \bottomrule
    \end{tabular}
\end{table}

 \setlength{\tabcolsep}{4pt}
\renewcommand{\arraystretch}{1.4}
\begin{table}[b!]
    \caption{Gate encoding in both bases using the same weights as in \autoref{tab:state-encoding}.}
    \label{tab:gate}
    \vspace{-1em}
     \label{tab:2}
\footnotesize
    \centering
\scalebox{0.95}{    \begin{tabular}{c|ccc}
    \toprule
        & \bf Comp. Basis (\combasis) & \bf Pauli Basis (\paulibasis) \\
    \midrule
    \makecell{1-qubit unitary ($U1$)} & $F_{U1}(\vec q, \vec q\,')$ where $\vec q = (q_0)$ & $F_{U1}(\vec q,\vec q\,')$ where $\vec q = (x,z)$  \\
    \makecell{$n$-qubit unitary ($Un$)} & $F_{\mathit{Un}}(\vec q, \vec q\,')$ where $\vec q = (q_0,\dots, q_n)$ & $F_{\mathit{Un}}(\vec q, \vec q\,')$ where $\vec q = (x_0,z_0,\dots, x_{n-1},z_{n-1})$  \\
    \midrule
        $H= \nicefrac1{\sqrt2} \smat{1&1\\1&-1}$ & $  F_H(q,q'\textcolor{gray}{,r,h} ) =  h \land (r \Leftrightarrow qq') $ 
            & $F_H(\vec q,\vec q\,'\textcolor{gray}{,r} ) =  (r \Leftrightarrow xz)
	\land  (z' \Leftrightarrow  x)
    \land  (x' \Leftrightarrow  z)$  \\
    $T= \smat{1&0\\0&\omega}$ & 
    $  F_T(q,q'\textcolor{gray}{,w} ) =  (q\Leftrightarrow q')\wedge(w \Leftrightarrow q)$
    & \textit{given in} \cite{mei2024simulating} \\
    $CX= \smat{1&0&0&0\\0&1&0&0\\0&0&0&1\\0&0&1&0}$ 
        &$(q_0'\Leftrightarrow q_0) \land (q_1'\Leftrightarrow (q_1 \oplus q_0))$ 
        & \textit{given in} \cite{mei2024simulating} \\
    \bottomrule
    \end{tabular}}
\end{table}

Applying a quantum gate maps a quantum state to another.
The encoding of a quantum gate $G$ is given by a Boolean function written as $F_G(\vec q, \vec q\,')$, 
where $\vec q$ is the input state and $\vec q'$ is output state after applying $G$, such that $F_{\ket{\psi}}(\vec q\, ') = F_{\ket{\phi}}(\vec q) \wedge F_G(\vec q, \vec q\,')$ for $\ket{\psi} = G\ket{\phi}$. As with state encodings, we introduce auxiliary variables when needed to represent weights introduced by the gates, and often omit these variables from the function signature.
We give an example of how to encode the gates $H$, $T$, and $CX$ in \autoref{tab:gate}.
A layer $D$ is encoded by conjoining the encodings of local gates, each applied to the variables of the relevant qubits.
For example, a two-qubit layer $D=\{H_0, T_1\}$ is encoded by $F_{D}(\q, \q\,') = F_H(\q_{0},\q\,'_{0})\wedge F_T(\q_{1},\q\,'_{1})$.

A quantum circuit $\mathcal C = (D^0,\dots, D^{d-1})$ is encoded by reserving variables $\vec q\,^0, \dots, \vec q\,^d$ for representing the initial, intermediate, and final states, and conjoining all the encoding of layers over these variables,
i.e., $F_{\mathcal C}(\vec q\,^0,\dots, \vec q\,^d) = \bigwedge_{j\in [d]}F_{D^j}(\vec q\,^j, \vec q\,^{j+1})$. When we don't need to name the intermediate states $\vec q\,^1,\dots, \vec q\,^{d-1}$, we will omit them, writing $F_{\mathcal C}(\vec q, \vec q\,') = F_{\mathcal C}(\vec q\,^0 = \vec q,\dots, \vec q\,^d = \vec q\,')$.

\autoref{lem:coefficient} shows that both WMC encodings allow for the strong simulation of any quantum circuit according to the usual definition of computing output amplitudes or coefficients~\cite{kissinger2022simulating}.
\begin{lemma}[\cite{mei2024gapp}]\label{lem:coefficient}
    Given an input state $\ket{\varphi} = \sum_{b\in{\{0,1\}}^n}\alpha_b\ket{b}$, such that $\ketbra{\varphi}{\varphi} = \sum_{j\in [4^n]}\beta_j\mathcal{P}_j$, an $n$-qubit circuit $\mathcal{C}$,
    a computational basis state $\ket{b}$ ($b\in\{0,1\}^n$) and a Pauli string $\mathcal P_j\in\{I,X,Y,Z\}^{\otimes n}$, the following equations hold:
    \[
    \begin{aligned}
        & \#SAT_W(F_{\ket{\varphi}}(\vec q)\wedge F_\mathcal{C}(\vec q, \vec q') \wedge F_{\ket{b}}(\vec q')) = \alpha_b & \text{in \ \combasis}, \\
        & \#SAT_W(F_{\ket{\varphi}}(\vec q)\wedge F_\mathcal{C}(\vec q, \vec q') \wedge F_{\mathcal P_j}(\vec q')) = \beta_j & \text{in \ \paulibasis}.
    \end{aligned}
    \]
\end{lemma}
We emphasize that strong simulation is canonical for quantum complexity classes~\cite{Thanos2024automated}, and therefore generalizes naturally to computing any measurement outcome probability~\cite{mei2024simulating} and determining circuit equivalence~\cite{mei2024eq}.
 \section{Problem Statement}
\label{sec:problem}

The quantum circuit synthesis problem seeks to construct a circuit that implements a given specification, which is provided as either a circuit or a unitary operator.
A key component is to determine if a guessed candidate circuit is \emph{equivalent} to the desired specification by checking unitary equivalence, as formalized in \autoref{def:approx_eq}.

\begin{definition}[Unitary equivalence]\label{def:approx_eq}
    Let $U$, $V$ be two $n$-qubit unitaries,
    Then $U$ and $V$ are $\epsilon$-equivalent,
    written as $U\simeq_{\epsilon} V$,
    if and only if the \jfid{} between $U$ and $V$ is not smaller than $1-\epsilon$, i.e. $\Fid_J(U,V) \geq 1-\epsilon$, where $\epsilon\in[0,1]$.
\end{definition}
In the above definition, 
we use the Jamio\l{}kowski fidelity defined in \autoref{Jfid} to measure the distance between two unitaries.
In particular, if and only if $\Fid_J(U,V) = 1$ ($\epsilon = 0$),
we have \emph{exact equivalence},
denoted as $U\equiv V$. In this case, $U$ and $V$ are equivalent up to a global phase $\lambda$ satisfying $|\lambda|^2 = 1$, i.e. $U = \lambda V$.
Building on equivalence checking,
\autoref{prob:c2c} gives the formal definition of (exact and approximate) quantum circuit synthesis.

\begin{problem}[Quantum circuit synthesis]\label{prob:c2c}
Given a specification represented by a circuit $\circuit_{1}$ in a gate set $\gateset_{1}$ or unitary $U_{\circuit_{1}}$ and an accuracy parameter $\epsilon \in (0,1]$,
the approximate synthesis problem asks for a depth-optimal quantum circuit $\circuit_2$ in a target gate set $\gateset_{2}$,
such that $U_{\circuit_1}\simeq_\epsilon  U_{\circuit_2}$.
For $\epsilon=0$, this is the exact synthesis problem.
\end{problem}

We consider Clifford+T as the elementary target gate set for synthesis because of its universality and importance in error-corrected quantum computing~\cite{fowler2012surface,aliferis2006quantum}.
While unitaries can be synthesized in Clifford+T (with additional ancillas) exactly if and only if their matrix entries belong to $\mathbb{Z}[\frac{1}{\sqrt{2}}, i]$~\cite{giles2013exact}, all unitaries can be synthesized approximately up to an arbitrarily small $\epsilon$~\cite{dawson2005solovaykitaev}.

In this work, we consider both exact and approximate synthesis and, for the first time, show that both problems can be reduced to MWMC. \section{Exact Quantum Circuit Synthesis}\label{sec:exactsyn}

This section presents our reduction from the exact quantum circuit synthesis problem to the MWMC problem.
We first provide an overview of the chosen approach.
Recalling \autoref{prob:c2c}, the exact synthesis problem has the following components:
\begin{itemize}
    \item \textbf{Input:} A quantum circuit $\circuit_1$ in gate set $\gateset_1$, or a unitary $U = U_{\circuit_1}$, and a finite target gate set $\gateset_2$. 
\item \textbf{Output:} A depth-optimal quantum circuit $\circuit_2$ in gate set $\gateset_2$, such that $U_{\circuit_2}\equiv U_{\circuit_1}$, if possible.
\end{itemize}
The problem can be expressed as exhaustively searching over all possible layers $\mathcal D_n$ for each layer in $\circuit_2$.
To achieve depth-optimality, we increment the depth $d$ until the following holds:
\begin{equation}\label{eq:exact_syn}
\exists D^0, \dots, D^{d-1}  \in \mathcal{D}_n  ~\colon~  U_{\circuit_1} \equiv U_{\circuit_2} \text{ where } \circuit_2 = (D^0, \dots, D^{d-1}). 
\end{equation}
In the remainder of this section, we first give the encoding of the input specification.
Then we present different encodings for checking the exact equivalence between the input and output.
Next, we encode a generic gate layer $\mathcal{D}_n$ and extend the encoding progressively to construct candidate circuits $\circuit_2$ of increasing depth, as specified in \autoref{eq:exact_syn}.
We conclude by showing how MWMC can find a depth-optimal output circuit~$\circuit_2$ using the encoding for generic layers and equivalence checking.

\textbf{Encoding input specifications.}
Due to the reversibility of quantum circuits, verifying the equivalence of circuits $\circuit_1$ and $\circuit_2$ is reducible to checking if the circuit $\circuit_2 \cdot  \circuit_1^\dagger$ is equivalent to the identity.
Therefore, we encode the input circuit $\circuit_1$ as its inverse $\circuit_1^\dagger$, denoted as $F_{\circuit_1^\dagger}(\q, \q\,')$, as explained in \autoref{sec:pre-quan}.
If the input is an $n$-qubit unitary operator $U^\dagger$, we can encode directly it as $F_{U^\dagger}(\vec q, \vec q\,')$ using weighted auxiliary variables to represent the unique components of the unitary, in the same way we encode individual gates~\cite{mei2024gapp}.

\textbf{Verifying exact equivalence.}
The exact equivalence checking problem is efficiently tackled with WMC in \cite{mei2024eq},
using a so-called \emph{linear encoding},
which requires $2n$ separate WMC calls.
In \autoref{prop:exact-eqcheck}, we extend this encoding to \emph{cyclic encoding} and \emph{linear-cyclic encoding}. These two new encodings solve the problem with a single call to the weighted model counter,
\emph{as we require in the proposed synthesis approach}.

\begin{restatable}{theorem}{exacteqcheck}
\label{prop:exact-eqcheck}
Let $\mathcal C$ be an $n$-qubit circuit, which is encoded by $F_{\mathcal C}$ with the corresponding weight function $W$.
Then, the following four statements are equivalent to each other:
\begin{itemize}
    \item $\circuit \equiv I^{\otimes n}$ ($\circuit$ is equivalent to the identity circuit $I^{\otimes n} = I_{2^n}$).
    \item Encoding the circuit $\circuit$ in \paulibasis, the  \textbf{linear encoding} has weighted model count~\cite{mei2024eq}:
    \begin{equation}\label{eq:linear-encoding}
    \wmc(F_{\mathcal{P}}(\vec{q}) \wedge F_\mathcal C\left(\vec{q}, \vec{q}\,' \right) \wedge F_{\mathcal{P}}(\vec{q}\,')) = 1 \text{ for all } \mathcal{P} \in \set{X_j, Z_j \mid j\in [n]}.
    \end{equation}
    \item Encoding the circuit in either \combasis{} or \paulibasis, the \textbf{cyclic encoding} has weighted model count \\
    (this approach can be viewed as checking ``overlap'' with the identity $I^{\otimes n}$):
    \begin{equation}\label{eq:cyclic-encoding}
    \wmc(F_\mathcal C\left(\vec{q}, \q\,' \right) \wedge F_{I^{\otimes n}}(\vec{q}, \vec{q}\,')) = c
      \text{ with } |c|=2^n \text{ for \combasis{} and }
    c = 4^n \text{ for \paulibasis}.
    \end{equation}
    \item Encoding the circuit $\circuit$ in \paulibasis, the \textbf{linear-cyclic encoding} has weighted model count:
    \begin{equation}\label{eq:linear-cyclic-encoding}
    \wmc(\bigvee_{\mathcal{P} \in \set{X_j, Z_j \mid j\in [n]}}F_{\mathcal{P}}(\vec{q}) \wedge F_\mathcal C\left(\vec{q}, \vec{q}\,' \right) \wedge F_{I^{\otimes n}}(\vec{q}, \vec{q}\,')) = 2n.
    \end{equation}
\end{itemize}
 where $\vec{q}, \vec{q}\,'$ are Boolean variables encoding the initial and final quantum state, respectively.
 Note that $ F_{I^{\otimes n}}(\vec q, \vec q\, ') = \bigwedge_{i\in [n]} ({q}_i\Leftrightarrow {q}_i')$, where ${q}_i\Leftrightarrow {q}_i'$ is shorthand for $(x_i\Leftrightarrow x_i') \land (z_i\Leftrightarrow z_i')$ in the Pauli basis.
\end{restatable}

We give a detailed proof in \aref{app:prof_exact}, and illustrate the theorem in \autoref{ex:equiv}.
\begin{example}
\label{ex:equiv}
Consider two circuits $\circuit_1= (S)$ and $\circuit_2= (T,T)$.
To check their equivalence, we first encode the circuits $\circuit = \circuit_2 \cdot \circuit_1^\dagger$ and $I$:
\[
\begin{aligned}
F_\circuit(q^0,  q^3) := F_{S^\dagger}( q^0, {q}^1) \wedge F_T(q^1, {q}^2) \wedge F_T(q^2, {q}^3)  \quad\text{  and  }\quad  F_I(q, q') := q \Leftrightarrow q'
\end{aligned}
\]
Then we check if $\circuit \equiv I$ as follows:\vspace{-1.5ex}
\begin{itemize}
    \item Linear encoding: $\wmc(F_\mathcal{P}(q^0) \land F_\circuit(q^0,  q^3) \land  F_\mathcal{P}(q^3)) = 1$  for $\mathcal{P} \in \set{X, Z}$ in \paulibasis.
\item Cyclic encoding: $\wmc(F_\circuit(q^0,  q^3) \wedge F_I(q^0,  q^3)) = c$ with
$|c| = 2$ (\combasis{}) or $c= 4$ (\paulibasis).
\item Linear-cyclic encoding: $\wmc((F_X \lor F_Z)\land  F_\circuit \wedge F_I) = 2$ in \paulibasis.\qed
\end{itemize}
\end{example}

\textbf{Encoding synthesis layers.}
Building up the output circuit is done by incrementally finding a sequence of gate layers implementing the input specification, as specified in \autoref{eq:exact_syn}.
To explore the space of possible gate layers, we encode the complete set $\mathcal{D}_n$ by introducing gate-selecting variables. Once these variables are fixed, they uniquely determine a specific gate layer within the set.

Given the target gate set $\mathcal{G}$,
we first consider single-qubit gates $\mathcal{G}(1)\subseteq \mathcal{G}$.
For each gate $G\in \mathcal G(1)$, we define the Boolean variable $p_{G,i}$ for each qubit $i\in [n]$,
where $p_{G,i}$ is true if and only if $G_i$ is included in the layer.
Thus, we can encode the single-qubit gates of the layer as:
\vspace{-1em}
\begin{equation}\label{cons:single}
F_{\mathcal G(1)}(\vec q, \vec q\,', \vec p(1)) = \bigwedge_{i\in[n]}\bigwedge_{G\in\mathcal{G}(1)} (p_{G,i}\Rightarrow F_{G}(q_i, q'_i)),
\end{equation}
where the variables $\vec q$ and $\vec q\,'$ encode the states before and after the layer respectively,
and $\vec p(1) = \{p_{G,i} \ | \ G\in\mathcal G(1), i\in[n]\}$ are the single-qubit gate-selecting variables.

Similarly, for a two-qubit gates $\mathcal G(2)\subseteq\mathcal G$,
we introduce quadratically many variables for all combinations of the two qubits: $\vec p(2) = \{p_{G,i,j}\ | \ G\in \mathcal G(2), i,j\in[n], i\neq j \}$. Then the encoding is:
\begin{equation}\label{cons:multi}
F_{\mathcal G(2)}(\vec q, \vec q\,', \vec p(2)) = 
\bigwedge_{i,j\in[n], j\neq i} \ \bigwedge_{G\in\mathcal{G}(2)} (p_{G,i,j}\Rightarrow F_{G_{i,j}}(\vec q, \vec q\,')),
\end{equation}

Since we consider the universal gate set Clifford+T,
which only includes single-qubit gates and two-qubit gates, the set of gate-selecting variables is given by $\vec p = \vec p(1) \cup \vec p(2)$ of the synthesis layer.
It is worth noting that the identity operator $I$ is included when considering multi-qubit circuits for completeness. 
For example when considering $\{CX,H,T\}$,
we define the set of single-qubit gates as $G(1) = \{H,T,I\}$.

A layer is valid only when applying exactly one gate to each qubit.
For example, $p_{H,1}$ and $p_{T,1}$ should not be true simultaneously.
Therefore, we define $\mathtt{EXO}(V) = \bigvee_{v\in V} v \wedge \bigwedge_{u,v\in V, u\neq v}(\overline{v} \vee \overline{u})$, a constraint ensuring that exactly one variable in the set $V$ is true.
Now, for each qubit $i \in [n]$, we apply:
\begin{equation}\label{cons:exo}
F_{\mathtt{EXO}}(\vec p) = \bigwedge_{i\in [n]}\mathtt{EXO}(\vec p_i),
\text{ where }
\vec p_i = \{p_{G,i} \ | \ G\in\mathcal G(1)\} \cup \{p_{G,i,j}, p_{G,j,i}\ | \ G\in \mathcal G(2), j\in[n] \}.
\end{equation}
Combining the three constraints in \autoref{cons:single}, \autoref{cons:multi} and \autoref{cons:exo} gives us the layer encoding:
\begin{equation}\label{eq:syn}
    F_{\mathcal{D}, \mathcal G}(\vec q, \vec q\,',\vec p) = 
F_{\mathcal G(1)}(\vec q, \vec q\,', \vec p(1)) \wedge
F_{\mathcal G(2)}(\vec q, \vec q\,', \vec p(2)) \wedge F_{\mathtt{EXO}}(\vec p)
\end{equation}

If there are gates in the target gate set $\gateset$ applied to more than two qubits, the encoding is expanded accordingly. 

\begin{example}
    The encoding of a layer in the universal gate set $\mathcal G = \{CX, H, T\}$ is as follows:
\[
\begin{aligned}
F_{\mathcal{D}, \mathcal G}(\vec q, \vec q\,', \vec p) = &
    \bigwedge_{i\in[n]}
    \bigl(p_{H,i} \Rightarrow F_H(q_i, q'_i) \wedge 
        p_{T,i} \Rightarrow F_T(q_i,q'_i) \wedge 
     p_{I,i} \Rightarrow F_I(q_i,q'_i)\bigr) \\ 
    ~~~\land &~~~ \bigwedge_{i,j\in[n], j \neq i}
    \bigl(p_{CX,i,j} \Rightarrow F_{CX_{i,j}}((q_i, q_j), (q'_i,q'_j))\bigr) ~~~\land~~~ \bigwedge_{i\in[n]} \mathtt{EXO}(\vec p_i),
   \\
\end{aligned}
\]
where $\vec p = \bigcup_{i\in[n]} \vec p_i$ and $\vec p_i = \{p_{I,i}, p_{H,i}, p_{T,i}\} \cup \bigcup_{j\in[n], j\neq i}\{p_{CX,i,j}, p_{CX,j,i}\}$.
\end{example}

To encode multiple layers,
we reserve state variables $\vec q\,^t$ and gate-selecting variables $\vec p\,^t$ for the $t$-th layer and substitute the above formula with the variables for each time step, i.e. $F_{\mathcal{D}, \mathcal G}(\vec q\,^{t}, \vec q\,^{t+1}, \vec p\,^t)$.
Encoding $d$ layers is given by a conjunction of the encodings: 
\begin{equation}\label{cons:gate-layer}
\bigwedge_{t\in[d]} F_{\mathcal{D}, \mathcal G}(\vec q\,^{t}, \vec q\,^{t+1}, \vec p\,^t).
\end{equation}

\textbf{Encoding exact synthesis.}
After giving the encoding of the input specification $U_{\circuit_1}$, the encoding of exact equivalence checking, and the encoding of gate layers, we now show the encoding of exact synthesis.
The main idea is to determine a minimal sequence of gate layers $\circuit_2 = (D^1, \dots, D^d)$ such that $U_{\circuit_2}\equiv U_{\circuit_1}$,
which can be represented with the cyclic encoding as: \vspace{-1ex}
\begin{equation}
\label{eq:syn_final_equation}
Syn_{\mathbf{C}, \mathbf{B}, \mathcal G,\mathcal C_1, d}(P,Q) ~~~=~~~
\overbrace{F_{U_{\mathcal C_1}^\dagger}\left(\vec{q}, \vec{q}\,^0 \right) 
\wedge  \bigwedge_{t\in[d]} F_{\mathcal{D}, \mathcal G}(\vec q\,^{t}, \vec q\,^{t+1}, \vec p\,^t)}^{U_{\mathcal C_2} U_{\mathcal C_1}^\dagger} 
~~\wedge ~~ \overbrace{(\vec q \Leftrightarrow \vec q\,^d)}^{F_{I^{\otimes n}}} ,
\end{equation}
where $\mathbf C$ denotes the encoding is cyclic, $\mathbf{B}$ denotes the chosen basis (which can be either \paulibasis or \combasis),
$P = \bigcup_{t\in[d]}\vec p\,^t$ is the set of gate-selecting variables, and
$Q = \vec q \cup \bigcup_{t\in[d+1]}\vec q\,^t \cup \aux\,^t$ is the set of all state variables $\vec q$ and auxiliary variables $\aux$. 
For a linear-cyclic checking (denoted by $\mathbf{LC}$)(only for \paulibasis),
the encoding is done by adding constraints to the initial state variables 
\begin{equation}
\label{eq:syn_final_equation2}
Syn_{\mathbf{LC}, \text{\paulibasis}, \mathcal G,\mathcal C_1, d}(P,Q) = \bigvee_{\pauli\in\{X_j,Z_j|j\in[n]\}}F_\pauli(\vec q)\wedge Syn_{\mathbf{C}, \text{\paulibasis}, \mathcal G,\mathcal C_1, d}(P,Q).
\end{equation}

\autoref{prop:syn-exact} shows how the exact synthesis problem is reduced to the MWMC problem.
It essentially reduces the problem of finding an assignment for the gate-selecting variables that maximizes the weighted model count of the above encodings in \autoref{eq:syn_final_equation} and \autoref{eq:syn_final_equation2}.
The correctness of \autoref{prop:syn-exact} follows directly from the cyclic and linear-cyclic encodings of \autoref{prop:exact-eqcheck},
which concludes this section.

\begin{proposition}\label{prop:syn-exact}
    Given a quantum circuit $\circuit_1$ (or its unitary $U_{\circuit_1}$) and an integer $d$,
    there exists a $d$-depth circuit $\circuit_2$ such that 
    $U_{\circuit_1}\equiv U_{\circuit_2}$ iff 
    \[
    \text{\maxwmc}\bigl(Syn_{\mathbf{E}, \mathbf{B}, \mathcal G,\mathcal C_1, d}(P,Q)\bigr) = (c, \tau(P)),
    \]
    where for $\mathbf{E}=\mathbf{LC}$ (linear-cyclic encoding) and $\mathbf{B} = \text{\paulibasis}$ we have that $c=2n$, and for $\mathbf{E} = \mathbf{C}$ (cyclic encoding) if $\mathbf{B} = \text{\paulibasis}$ we have that $c = 4^n$ and if $\mathbf{B} = \text{\combasis}$ we have that $|c|=2^n$.
    When the maximal value is achieved, the output circuit $\mathcal C_2$ can be directly determined from the satisfying assignment $\tau(P)$.
\end{proposition}

\begin{example}
Let us consider the circuit $\mathcal C_1= (S)$ and the target gate set  $\mathcal G =\{CX, H, T\}$. Since it is a single-qubit circuit, the encoding of one general layer is, where for gate $G$ we denote $F_G=F_G(q_0,q'_0)$
\[
\begin{aligned}
F_{\mathcal{D}, \mathcal G}(\vec q = (q_0), \vec q\,' = (q'_0), \vec p) =
    \bigl((p_{H,0} \Rightarrow F_H) \wedge 
        (p_{T,0} \Rightarrow F_T) \wedge 
     (p_{I,0} \Rightarrow F_I&)\bigr)
    ~\land ~ \bigwedge_{i\in[n]} \texttt{EXO}( \vec p ), \end{aligned}
\]
We synthesize the circuit by calling $Max\#SAT_W$ to find an assignment for the gate-selection variables such that the function achieves maximal value:
\[
\begin{aligned}
Max\#SAT_W(Syn_{\mathbf{LC}, \text{\paulibasis}, \mathcal G,\mathcal C_1, 1}(P,Q)) &= (0.854, ~ \tau_1=\{ p_{H,0}^0 \leftarrow 0, p_{T,0}^0\leftarrow 1, p_{I,0}^0\leftarrow 0\}) \\ 
Max\#SAT_W(Syn_{\mathbf{LC}, \text{\paulibasis}, \mathcal G,\mathcal C_1, 2}(P,Q)) &= (1, ~ \tau_2=\{ p_{H,0}^0 \leftarrow 0, p_{T,0}^0 \leftarrow 1,  p_{I,0}^0 \leftarrow 0, \\ & \quad\quad\quad\quad\quad
 p_{H,0}^1 \leftarrow 0, p_{T,0}^1 \leftarrow 1,  p_{I,0}^1 \leftarrow 0
\})
\end{aligned}
\]
From the assignment $\tau_2$, we get the synthesized circuit $\mathcal C_2= (T, T)$.
\end{example}

 \section{Approximate Quantum Circuit Synthesis}
In this section,
we focus on approximate synthesis as given in \autoref{prob:c2c}.
The input and output of the problem are now as follows:
\begin{itemize}
    \item \textbf{Input:} A quantum circuit $\circuit_1$ in a gate set $\gateset_1$ or directly the unitary $U_{\circuit_1}$, a finite target gate set $\gateset_2$, and an error bound $\epsilon\in(0,1]$.
    \item \textbf{Output:} A depth-optimal quantum circuit $\circuit_2$ in gate set $\gateset_2$ such that $\circuit_2\simeq_{\epsilon}\circuit_1$.
\end{itemize}
Since synthesis relies on equivalence checking, our first aim is to lift the latter to 
approximate equivalence checking.
\autoref{prop:fidwmc} shows that the cyclic encoding in both \paulibasis{} and \combasis{} computes the \jfid{} between two circuits, thus both encodings can determine approximate equivalence checking based on \autoref{def:approx_eq}.

\begin{restatable}{theorem}{fidelity}\label{prop:fidwmc}
Given two unitary matrices $U$ and $V$ on an $n$-qubit Hilbert space $\mathcal{H}^{\otimes{n}}$,
    the \jfid{} can be computed by
    \[
    \Fid_J (U,V)
    ~~~=~~~ 
    \begin{cases}
    \frac{1}{4^n} \cdot     {\#SAT_W(F_{U^\dagger V}(\vec q, \vec q\,')\wedge (\vec q\Leftrightarrow \vec q\,'))} & \text{in } ~~ \paulibasis \\
    \frac{1}{4^n} \cdot \abs{\#SAT_W(F_{U^\dagger V}(\vec q, \vec q\,')\wedge (\vec q\Leftrightarrow \vec q\,'))}^2 & \text{in } ~~ \combasis
    \end{cases}.
\]
\end{restatable}
\begin{proof}
    To compute \jfid{} of two given $n$-qubit circuits $U$ and $V$ on $\mathcal{H}^{\otimes n}$,
one takes the maximally entangled state $\ket{\Psi_n}$ on $\mathcal{H}^{\otimes n}\otimes\mathcal{H}^{\otimes n}$ as input and compute the fidelity between the output states $(U\otimes I^{\otimes n}) \ket{\Psi_n}$ and $(V\otimes I^{\otimes n}) \ket{\Psi_n}$.
We first prove this in the \combasis{} and then move to \paulibasis.

In \combasis{}, the \jfid{} is given by
\[
\begin{aligned}
\Fid_J(U,V) &= \Fid((U\otimes I^{\otimes n})\ket{\Psi_n}, (V\otimes I^{\otimes n})\ket{\Psi_n}) =
\left|\bra{\Psi_n}(U^\dagger\otimes I^{\otimes n})(V\otimes I^{\otimes n})\ket{\Psi_n}\right|^2\\
& = \left|\frac{1}{2^n}\cdot\sum_{b\in \{0,1\}^n}\sum_{b'\in \{0,1\}^n}\bra{i}U^\dagger V\ket{j}\cdot \bra{b}I^{\otimes n}\ket{b'}\right|^2 =
\left|\frac{1}{2^n}\cdot\sum_{b\in \{0,1\}^n}\bra{b}U^\dagger V\ket{b}\right|^2,\\
\end{aligned}
\]
so $\Fid_J(U,V) = \tr(U^\dagger V)$.
Based on \autoref{lem:coefficient}, it holds that:
\[
\begin{aligned}
\#SAT_W(F_{U^\dagger V}(\vec q, \vec q\,')\wedge(\vec q\Leftrightarrow \vec q\,')) &= \sum_{b\in \{0,1\}^n} \#SAT_W(F_{\ket{b}}(\vec q) \wedge F_{U^\dagger V}(\vec q, \vec q\,') \wedge F_{\ket{b}}(\vec q\,')) \\ &=
\sum_{b\in \{0,1\}^n}\bra{b}U^\dagger V\ket{b},
\end{aligned}
\]
giving us that $\Fid_J(U,V) = \frac{1}{4^n}\left| \#SAT_W(F_{U^\dagger V}(\vec q, \vec q\,')\wedge(\vec q\Leftrightarrow \vec q\,'))\right|^2$ as we wanted to show.

Moving to \paulibasis{}, we first represent $\ket{\Psi_n}$ and the \jfid{} in the Pauli basis, and then explain the encoding.
The density operator of the maximally entangled state can be decomposed in the Pauli basis as:
$
\ketbra{\Psi} = 
\frac{1}{2^n}\sum_{i\in[n],j\in[n]}\ketbra{ii}{jj} = \frac{1}{4^n}\sum_{\mathcal{P}_i\in\{X,Y,Z,I\}^{\otimes n}} \mathcal{P}_i\otimes \mathcal{P}^T_i,
$
which is shown in \cite{chen2021quantumaf}.
The \jfid{} can be defined in the Pauli basis as:
    \begin{equation}\label{eq:fidpauli}
    \begin{aligned}
        \Fid_J&(U^\dagger V,I) = 
        \Tr((U^\dagger V\otimes I)\ketbra{\Psi}(V^\dagger U\otimes I)\ketbra{\Psi}) \\
        & = \frac{1}{16^n}\sum_{j,k\in[4^n]}\Tr(U^\dagger V\mathcal{P}_j V^\dagger U \mathcal{P}_k \otimes \mathcal{P}^T_j\mathcal{P}^T_k) = \frac{1}{8^n}\sum_{j\in[4^n]}\Tr(U^\dagger V\mathcal{P}_j V^\dagger U \mathcal{P}_j)
    \end{aligned}
    \end{equation}
    Based on Proposition 1 in \cite{mei2024eq},
    each of the summand $\frac{1}{2^n}\cdot\Tr(U^\dagger V\mathcal{P}_j V^\dagger U \mathcal{P}_j)$ can be computed by the weighted model counting of $F_{\mathcal{P}_j}(\vec q)\wedge F_{U^\dagger V}(\vec q, \vec q\,')\wedge F_{\mathcal{P}_j}(\vec q\,')$,
    or equally,
    $F_{\mathcal{P}_j}(\vec q)\wedge F_{U^\dagger V}(\vec q, \vec q\,')\wedge (\vec q \Leftrightarrow \vec q\,')$.
    To compute the summation, 
    one should go over all possible Pauli strings,
    which is equal to setting the variables in $\vec q$ to be free.
    Hence, the Jamiolkowski fidelity is obtained by $\Fid_J(U,V) = \frac{1}{4^n}\#SAT_W(F_{U^\dagger V}(\vec q\,^0, \vec q\,^m)\wedge (\vec q\,^0 \Leftrightarrow \vec q\,^m))$.
\end{proof}
Let $\circuit_1$, $\circuit_2$ be two circuits,
since $\Fid_J(\circuit_1,\circuit_2) = \Tr(\circuit_1^\dagger \circuit_2) = \Tr(\circuit_2\circuit_1^\dagger) = \Fid_J(I, \circuit_2\circuit_1^\dagger)$, 
and based on \autoref{prop:fidwmc}, one can reuse   \autoref{eq:syn_final_equation} to compute $\Fid_J(\circuit_1, \circuit_2)$.

\begin{proposition}\label{prop:approx-syn}
    Given a quantum circuit $\circuit_1$ on $n$ qubits, an integer $d$ and an error bound $\epsilon\in(0,1]$,
    there exists a $d$-depth $n$-qubit circuit $\circuit_2$ such that $U_{\circuit_2}\simeq_\epsilon U_{\circuit_1}$ iff
    $
    Max\#SAT_W\bigl(Syn_{\mathbf{C}, \mathbf{B}, \mathcal G,\mathcal C_1, d}(P,Q)\bigr) = (c, \tau(P)),
    $
    where $\frac{1}{2^n}|c| > 1-\epsilon$ if $\textbf B = \combasis{}$ and $\frac{1}{4^n}c>1-\epsilon$ if $\textbf B = \paulibasis$.
    The circuit $\mathcal C_2$ can be determined by the satisfying assignment $\tau(P)$.
\end{proposition}
Like in exact synthesis, to get the depth-optimal $\epsilon$-approximate circuit,
we apply the above \autoref{prop:approx-syn} for an increasing depth $d$. \autoref{ex:fidelity} demonstrates approximate synthesis.

\begin{example}\label{ex:fidelity}
Consider the circuit $\mathcal C_1 = (R_Z(\frac{\pi}{8}))$ with $\epsilon = 0.05$.
Synthesizing one gate layer, we get the result as
\[
Max\#SAT_W(Syn_{\mathbf{B}, \paulibasis,\mathcal G,\mathcal C_1, 1}(P,Q)) = (3.848, \{ p_{H,0}^0\leftarrow 0, p_{T,0}^0\leftarrow 1, p_{I,0}^0\leftarrow 0\}),
\]
which determine the output circuit as $\circuit_2 = (T)$ and the fidelity $\Fid_J(U_{\circuit_1}, U_{\circuit_2}) = \frac{1}{4}\times 3.848 = 0.962 > 1- \epsilon$.
Thus, the synthesis procedure stops, and the output circuit is $(T)$.
\end{example}
 \section{Related Work}
\label{sec:related}

\textbf{Clifford circuit synthesis.}
Synthesis of Clifford circuits is substantially simpler than that of universal quantum circuits due to the ability to exploit the algebraic structure of the symplectic group, which significantly constrains the search space. Any Clifford operation on $n$ qubits can be efficiently represented by a $2n \times 2n$ symplectic matrix over $\mathbb{F}_2$, rather than the exponentially larger $2^n \times 2^n$ unitary matrix typically required for general quantum operations \cite{calderbank1997qec}.
Building on this, Maslov and Roetteler~\cite{maslov2018bruhat} employ the Bruhat decomposition of the symplectic group to generate shorter Clifford circuits. 
Similarly, Rengaswamy et al.~\cite{rengaswamy2018synthesis} develop Clifford synthesis algorithms via symplectic geometry, targeting logical-level Clifford operations with an emphasis on practical implementations on physical qubits.
While these approaches produce efficient Clifford circuits, they do not guarantee optimality in terms of depth or gate count. 
To address this limitation, Schneider et al.~\cite{schneider2023sat} reformulate Clifford synthesis as a satisfiability problem.
This encoding enables the use of SAT and MaxSAT solvers~\cite{biere2009handbook} to identify optimal Clifford circuits for a fixed depth, offering a rigorous method to achieve minimality.

\textbf{Exact Clifford+T circuit synthesis.}
In error-corrected quantum computing, the relevant universal gate is Clifford+T.
There are many works considering the exact synthesis of quantum circuits in the gate set Clifford+T\cite{amy2013mitm,matteo2016parallel,giles2013exact,gosset2013algorithmtcount,kliuchnikov2013fast,niemann2020advanced}, i.e., the desired specification is realized without any rounding errors.
Approaches like \cite{giles2013exact,niemann2020advanced} synthesize the unitary matrix representing the specification in a local fashion, i.e., column by column.
However, they do not give the optimal solution and leave significant room for improvement.
To achieve optimality, 
the meet-in-the-middle algorithm~\cite{matteo2016parallel,gosset2013algorithmtcount} performs an exhaustive search over the space of all Clifford+T circuits up to a given depth.

\textbf{Approximate Clifford+T circuit synthesis.}
Since not all unitaries can be implemented exactly in the Clifford+T gate set,
other works have focused on synthesizing circuits under different approximation metrics. Most of these works \cite{kliuchnikov2016practical,ross2016ancillafree,selinger2014efficient,fowler2011steane} focus on single-qubit operators, especially rotation gates, while \cite{gheorghiu2022tcount,paradis2024synthetiq} consider multi-qubit operators. \section{Experimental Evaluation}

\subsection{Implementation}
To evaluate and test our proposed method, we implemented it in a tool called \texttt{Quokka\#-syn}, as part of the open source toolkit \texttt{Quokka\#}\footnote{\url{https://github.com/System-Verification-Lab/Quokka-Sharp}}.
In addition to the core method described in the previous sections, we also implemented several optimization rules. 
Since no existing tool can solve our problem, we extended the MWMC solver d4Max~\cite{audemard2022maxsat} with support for negative and complex weights. The following paragraphs describe these two components in more detail.

\textbf{Encoding optimization.}
While the encoding in \autoref{cons:gate-layer} allows for any circuit of depth $d$, many encoded circuits are redundant as they fall in the same equivalence class.
For example, two consecutive Hadamard gates can be reduced to the identity.
So if the $H$ is applied to the $j$-th qubit at depth $t$,
we can safely exclude the case where another $H$ is applied to the same qubit at depth $t+1$. 
In the encoding, this corresponds to enforcing that if $p_{H_{j,t}}$ is set to \true, then $p_{H_{j,t+1}}$ must be \false,
which can be encoded as $\bigwedge_{k\in[d-1]}\bigwedge_{i\in[n]}(\overline{p}_{H,i}^{k}\vee\overline{p}_{H,i}^{k+1})$.
We apply similar reasoning to patterns like $T^8 = I$ and $CX_{i,j}CX_{i,j} = I^{\otimes 2}$, which are never part of an optimal circuit. Eliminating such cases does not affect optimality and helps reduce the search space. We introduce additional encodings to prune such redundant structures, as detailed in \aref{app:opt}.

\textbf{d4Max extension.} To support the encodings, we first extended d4Max to support negative weights (for \paulibasis) and complex weights (for \combasis). For number representation, we use arbitrary precision arithmetic from the GMP library, as in the original version of d4Max.
Because of the negative weights, computing an upper bound for a sub-formula becomes more complex.
In the original version of d4Max, upper bounds can be easily computed by assuming the formula is a tautology. This allows the solver to prune branches of the search tree when it is clear that no better value can be reached than the current best solution. As this is non-trivial to resolve in the presence of negative weights, we turned this optimization off.
Another removed feature is the ability to compute intermediate approximations before processing all connected components. As a result of these changes, the current version is slower at providing intermediate solutions compared to the original version.

\subsection{Performance}
In this section, we explore the performance of our method as implemented in \texttt{Quokka\#-syn}. Our implementation synthesizes circuits with the gate set $\{H, CX, T, T^\dagger\}$. This set is also universal since $S=TT$,
where we replace the $S$ gate with the $T^\dagger$ gate.

We demonstrate the feasibility and scalability of our method based on two classes of benchmarks for exact synthesis: random circuits on $\{H, CX, T, T^\dagger\}$ and commonly used unitary operators such as the $CH$ and $\toffoli$ gates.
We use randomly generated circuits as benchmarks, as they are the standard method for evaluating performance since they represent a hard case, which is also employed in practice to prove quantum supremacy \cite{arute2019quantum}.

Since our method guarantees a depth-optimal result, we compare with the state-of-the-art method \texttt{mitms}~\cite{amy2013mitm}, which targets the same task.
Our experiments were run on a single-core AMD Ryzen 9 7900X Processor and 64 GB of memory.

\textbf{Exact synthesis.}
In the first experiment, we compare the performances of the different encodings to determine the best one and assess the scalability of the method, as this determines its effectiveness in practice~\cite{amy2013mitm}. We test all three synthesis encodings for circuit inputs: \combasis with the cyclic encoding and \paulibasis with the cyclic and linear-cyclic encodings.

\begin{table}[t!]
    \caption{Synthesis benchmarks for random circuits. 
    \combasis and \paulibasis indicate cyclic encoding in each basis; L\paulibasis uses linear-cyclic encoding in \paulibasis. Memory is reported by d4max; for short runtimes, usage is below 2 GB and denoted by $<2$. Averages include $\pm$ standard deviation. \textbf{-} indicates all 10 samples failed; $\circ$ marks untested depths due to low prior success.}
    \label{tab:syn_statistics}
    \centering
\scalebox{.6}{
\begin{tabular}{ll|lll|lll|lll}
\toprule
\multirow[c]{2}{*}{\#Qb} & \multirow[c]{2}{*}{Depth} & \multicolumn{3}{c}{Rate} & \multicolumn{3}{c}{Time (s)} & \multicolumn{3}{c}{Memory (GB)} \\
 &  & \combasis & L\paulibasis & \paulibasis & \combasis & L\paulibasis & \paulibasis & \combasis & L\paulibasis & \paulibasis \\
\midrule
\multirow[t]{8}{*}{2} & 1 & 1.0 & 1.0 & 1.0 & $0.05\pm0.00$ & $0.05\pm0.00$ & $0.05\pm0.00$ & $<2$ & $<2$ & $<2$ \\
 & 2 & 1.0 & 1.0 & 1.0 & $0.07\pm0.00$ & $0.08\pm0.00$ & $0.08\pm0.00$ & $<2$ & $<2$ & $<2$ \\
 & 3 & 1.0 & 1.0 & 1.0 & $0.11\pm0.01$ & $0.11\pm0.01$ & $0.14\pm0.02$ & $<2$ & $<2$ & $<2$ \\
 & 4 & 1.0 & 1.0 & 1.0 & $0.29\pm0.09$ & $0.22\pm0.07$ & $0.40\pm0.10$ & $<2$ & $<2$ & $<2$ \\
 & 5 & 1.0 & 1.0 & 1.0 & $2.46\pm1.51$ & $1.46\pm0.72$ & $4.53\pm3.16$ & $2.13\pm0.03$ & $2.09\pm0.00$ & $2.11\pm0.02$ \\
 & 6 & 1.0 & 1.0 & 1.0 & $21.25\pm12.54$ & $14.26\pm10.33$ & $55.69\pm42.54$ & $2.46\pm0.28$ & $2.16\pm0.07$ & $2.38\pm0.27$ \\
 & 7 & 0.9 & 0.9 & 0.3 & $115.31\pm64.99$ & $97.28\pm46.98$ & $206.37\pm17.36$ & $3.84\pm1.31$ & $2.48\pm0.21$ & $3.41\pm0.54$ \\
 & 8 & 0.0 & 0.0 & $\circ$ & \textbf{-} & \textbf{-} & $\circ$ & \textbf{-} & \textbf{-} & $\circ$ \\
\cline{1-11}
\multirow[t]{6}{*}{3} & 1 & 1.0 & 1.0 & 1.0 & $0.05\pm0.00$ & $0.05\pm0.00$ & $0.05\pm0.00$ & $<2$ & $<2$ & $<2$ \\
 & 2 & 1.0 & 1.0 & 1.0 & $0.08\pm0.01$ & $0.10\pm0.01$ & $0.10\pm0.02$ & $<2$ & $<2$ & $<2$ \\
 & 3 & 1.0 & 1.0 & 1.0 & $1.43\pm0.97$ & $0.44\pm0.24$ & $2.45\pm1.91$ & $2.11\pm0.02$ & $<2$ & $2.11\pm0.01$ \\
 & 4 & 1.0 & 1.0 & 1.0 & $51.56\pm57.11$ & $11.26\pm4.95$ & $105.40\pm64.98$ & $3.15\pm1.25$ & $2.15\pm0.04$ & $2.87\pm0.67$ \\
 & 5 & 0.1 & 0.7 & 0.0 & $145.96$ & $147.37\pm49.30$ & \textbf{-} & $5.08$ & $3.20\pm0.53$ & \textbf{-} \\
 & 6 & $\circ$ & 0.0 & $\circ$ & $\circ$ & \textbf{-} & $\circ$ & $\circ$ & \textbf{-} & $\circ$ \\
\cline{1-11}
\multirow[t]{4}{*}{4} & 1 & 1.0 & 1.0 & 1.0 & $0.05\pm0.00$ & $0.06\pm0.00$ & $0.05\pm0.00$ & $<2$ & $<2$ & $<2$ \\
 & 2 & 1.0 & 1.0 & 1.0 & $0.37\pm0.20$ & $0.30\pm0.14$ & $0.62\pm0.48$ & $2.09$ & $<2$ & $2.09\pm0.01$ \\
 & 3 & 1.0 & 1.0 & 0.5 & $102.55\pm85.73$ & $22.65\pm16.03$ & $85.70\pm42.76$ & $4.44\pm2.07$ & $2.26\pm0.13$ & $2.69\pm0.34$ \\
 & 4 & 0.0 & 0.2 & 0.0 & \textbf{-} & $246.80\pm21.38$ & \textbf{-} & \textbf{-} & $4.51\pm0.20$ & \textbf{-} \\
\cline{1-11}
\multirow[t]{3}{*}{5} & 1 & 1.0 & 1.0 & 1.0 & $0.05\pm0.00$ & $0.06\pm0.01$ & $0.06\pm0.00$ & $<2$ & $<2$ & $<2$ \\
 & 2 & 1.0 & 1.0 & 1.0 & $2.93\pm2.17$ & $3.01\pm1.62$ & $11.85\pm11.15$ & $2.14\pm0.04$ & $2.11\pm0.02$ & $2.18\pm0.10$ \\
 & 3 & 0.0 & 0.1 & 0.0 & \textbf{-} & $243.71$ & \textbf{-} & \textbf{-} & $4.67$ & \textbf{-} \\
\cline{1-11}
\multirow[t]{3}{*}{6} & 1 & 1.0 & 1.0 & 1.0 & $0.05\pm0.00$ & $0.12\pm0.06$ & $0.06\pm0.00$ & $<2$ & $<2$ & $<2$ \\
 & 2 & 0.9 & 1.0 & 0.7 & $70.49\pm82.13$ & $54.47\pm44.75$ & $72.71\pm69.85$ & $3.74\pm2.03$ & $2.69\pm0.61$ & $2.80\pm0.81$ \\
 & 3 & 0.0 & 0.0 & 0.0 & \textbf{-} & \textbf{-} & \textbf{-} & \textbf{-} & \textbf{-} & \textbf{-} \\
\cline{1-11}
\bottomrule
\end{tabular}
}
\end{table}

 Random circuits are generated layer by layer, with each layer applying exactly one gate per qubit. Gates are selected uniformly at random from the gate set, excluding the identity gate, and assigned to unassigned qubits. Since a random circuit of depth $d$ could be equivalent to a shallower circuit, we retain only those circuits that cannot be synthesized with a depth less than $d$, until we obtain 10 samples. 
Using these samples, we then evaluated all three encodings for increasing depths, stopping when fewer than half the circuits are solved within the time limit per instance. We set the time limit to 300 seconds to accommodate a large number of benchmarks. 
Experiments were run on circuits with 2 to 6 qubits; the 1-qubit case is omitted here as it is evaluated for approximate synthesis below.

The results are shown in \autoref{tab:syn_statistics}. For each encoding, qubit count, and depth, we report the success rate (the fraction of random circuit samples that can be solved within the time limit), the average runtime of solved cases, and the average memory needed by d4Max for the solved cases.
The results highlight the significant impact of the qubit count and depth on the runtime. For every number of qubits, there is a specific depth threshold where problem-solving becomes challenging; before reaching this point, the success rate is perfect.
We observe that the linear encoding consistently demonstrates the best performance among the three encoding methods.

\renewcommand{\arraystretch}{1.1}
\begin{table}[t!]
\caption{Experimental results on synthesis for controlled-gates. We present the running time and memory usage reported by d4Max, as the encoding time is less than 0.01 seconds. We report N/A when d4Max does not give the memory usage, denoted by $<2000$. Here '-' denotes timeout cases except for the case \texttt{mitms} with 2-qubit and 8-depth, where the data is not given in \cite{amy2013mitm}.}
\label{tab:unitary}
\centering
\scalebox{0.6}{
\begin{tabular}{|llcl|r|r|r|r|r|r|r|r|}
\hline
\multicolumn{4}{|l|}{\#Qubits \textbackslash Depth}                                                                                                                                          & 1     & 2     & 3       & 4        & 5         & 6       & 7        & 8         \\ \hline
\multicolumn{1}{|l|}{\multirow{6}{*}{2}} & \multicolumn{1}{l|}{\multirow{4}{*}{\texttt{Quokka\#-syn}}} & \multicolumn{1}{l|}{\multirow{2}{*}{\begin{tabular}[c]{@{}l@{}}Cyclic \\ (unitary)\end{tabular}}} & Time (s) & 0.023 & 0.025 & 0.059   & 0.357    & 3.890      & 29.861  & 436.450   & -         \\ \cline{4-12} 
\multicolumn{1}{|l|}{}                    & \multicolumn{1}{l|}{}                            & \multicolumn{1}{l|}{}                                                                             & Mem (MB) & $<2000$   & $<2000$   & $<2000$     & $<2000$      & 2200.990   & 2767.230 & 10666.570 & -         \\ \cline{3-12} 
\multicolumn{1}{|l|}{}                    & \multicolumn{1}{l|}{}                            &\multicolumn{1}{l|}{\multirow{2}{*}{\begin{tabular}[c]{@{}l@{}}Linear-\\cyclic \end{tabular}}}                                  & Time (s) & 0.023 & 0.028 & 0.053   & 0.270     & 2.337     & 6.470    & 197.810   & 2353.696 \\ \cline{4-12} 
\multicolumn{1}{|l|}{}                    & \multicolumn{1}{l|}{}                            & \multicolumn{1}{l|}{}                                                                             & Mem (MB) & $<2000$   & $<2000$   & $<2000$     & $<2000$      & 2148.300   & 2183.760 & 4814.550  & 16091.660  \\ \cline{2-12} 
\multicolumn{1}{|l|}{}                    & \multicolumn{2}{l|}{\multirow{2}{*}{\texttt{\mitms}}}                                                                                          & Time (s) & 0.002 & 0.019 & 0.188   & 1.602   & 12.433    & 84.622  & -        & -         \\ \cline{4-12} 
\multicolumn{1}{|l|}{}                    & \multicolumn{2}{l|}{}                                                                                                                                & Mem (MB) & 0.002 & 0.016 & 0.147   & 1.013    & 6.249     & 32.766  & -        & -         \\ \hline
\multicolumn{1}{|l|}{\multirow{6}{*}{3}} & \multicolumn{1}{l|}{\multirow{4}{*}{\texttt{Quokka\#-syn}}} & \multicolumn{1}{l|}{\multirow{2}{*}{\begin{tabular}[c]{@{}l@{}}Cyclic \\ (unitary)\end{tabular}}} & Time (s) & 0.024 & 0.106 & 3.310    & 177.220   & -         & -       & -        & -         \\ \cline{4-12} 
\multicolumn{1}{|l|}{}                    & \multicolumn{1}{l|}{}                            & \multicolumn{1}{l|}{}                                                                             & Mem (MB) & $<2000$   & $<2000$   & 2188.920 & 5398.730  & -         & -       & -        & -         \\ \cline{3-12} 
\multicolumn{1}{|l|}{}                    & \multicolumn{1}{l|}{}                            & \multicolumn{1}{l|}{\multirow{2}{*}{\begin{tabular}[c]{@{}l@{}}Linear-\\cyclic \end{tabular}}}                                 & Time (s) & 0.029 & 0.140  & 3.413   & 76.293   & 3066.258 & -       & -        & -         \\ \cline{4-12} 
\multicolumn{1}{|l|}{}                    & \multicolumn{1}{l|}{}                            & \multicolumn{1}{l|}{}                                                                             & Mem (MB) & $<2000$   & $<2000$   & 2148.740 & 2688.360  & 18860.120  & -       & -        & -         \\ \cline{2-12} 
\multicolumn{1}{|l|}{}                    & \multicolumn{2}{l|}{\multirow{2}{*}{\texttt{\mitms}}}                                                                                          & Time (s) & 0.027 & 1.409 & 53.238    & 2311.023 & -         & -       & -        & -         \\ \cline{4-12} 
\multicolumn{1}{|l|}{}                    & \multicolumn{2}{l|}{}                                                                                                                                & Mem (MB) & 0.005 & 0.179 & 6.737   & 215.970   & -         & -       & -        & -         \\ 
\hline
\end{tabular}
}
\end{table}
 
The second experiment compares the performance of \texttt{Quokka\#-syn} with both unitary and circuit inputs to that of \texttt{mitms}, a tool that performs the same task,
only with target gate set $\{H,CNOT, S, S^\dagger, T, T^\dagger\}$.
For comparison, we selected the 2-qubit $CH$ gate and the 3-qubit $\toffoli$ gate. We encode the gates once as unitary matrices in \combasis, as described in \autoref{sec:exactsyn}, and once as Clifford+T circuits, as described in~\cite{nielsen2000quantum}, in \paulibasis.
We choose to use \combasis encoding for unitary inputs, since even for sparse unitaries, such as the Toffoli gate and its generalizations, the encoding of the Pauli decomposition may blow up, as shown in~\cite{mei2024eq}.
For the circuit encoding, we use the \emph{linear-cyclic encoding} in \paulibasis, as it showed the best performance in the previous experiment.
For both encodings, the timeout is set to one hour (including the encoding time and calling of d4Max for all depths). 
Since we failed to compile \texttt{mitms} within our experimental setup due to its reliance on older libraries, we instead refer to the performance data reported in~\cite{amy2013mitm} to provide a comparative reference.
The results are shown in \autoref{tab:unitary}. While the comparison is not entirely fair ---due to differences in platform, target gate set, and benchmark design--- we observe that \texttt{Quokka\#-syn} tends to exhibit higher memory usage but lower runtime compared to \texttt{mitms}. 

It is also worth noting that the performance of \paulibasis with the linear-cyclic encoding given in \autoref{tab:syn_statistics} are better than \autoref{tab:unitary} in many cases, for example, $3$-qubit and $5$-depth cases in \autoref{tab:syn_statistics} feature 70\% success rate with average runtime of 147.37 seconds, while in \autoref{tab:unitary}, it takes 3066.26 seconds.
One explanation is that, since the maximal weight is known, as shown in \autoref{prop:exact-eqcheck}, 
\texttt{Quokka\#-syn} provides a threshold to the solver, allowing it to terminate as soon as the optimal value is achieved.
Consequently, if a circuit can be successfully synthesized within the given depth, as with the solved cases in \autoref{tab:syn_statistics}, the tool completes early. In contrast, for unsuccessful cases, such as for each depth in \autoref{prop:exact-eqcheck},
the solver must explore the entire search space, resulting in longer run times.
In addition, we observe substantial variability in the run times, indicating that the performance is highly dependent on the specific characteristics of each case.

\begin{wraptable}{r}{0.55\textwidth}
\centering
\vspace{-1em}
\caption{$R_Z(\pi/8)$ synthesis using \combasis with cyclic encoding. Statistics are shown per synthesis layer. Memory usage under short runtimes is omitted ($<2$). Only fidelity-improving iterations are shown. At 24 layers, the program crashed due to resource limits (\textbf{-}).}
\label{tab:synRZ}

\scalebox{0.7}{
\begin{tabular}{l||r|r|r|r}
\toprule
\bf depth    & 1 & 10 & 15 & 24\\                                         
\midrule
\bf \#variables             & 11 & 92 & 137 & 218 \\  
\hline
\bf \#clauses               & 37 & 382 & 577 & 928 \\  
\hline
\bf \#literals              & 90 & 960 & 1460 & 2360 \\  
\hline
\bf \#Selecting variables   & 4 & 40 & 60 & 96 \\
\hline
\bf fidelity            & 0.962 & 0.975 & 0.997 & \textbf{-} \\
\hline
\bf Time (s)            & 0.021 & 0.092 & 2.685 & >1560.51 \\   
\hline
\bf mem (GB)            & <2  & <2 & 2.2 & >30 \\
\bottomrule
\end{tabular}
}
\end{wraptable}
 \textbf{Approximate synthesis.}
Important gates that often need to be synthesized approximately in Clifford+T are the general rotation gates: $R_x$, $R_y$, $R_z$~\cite{kliuchnikov2016practical}.
Our method supports a general rotation angle in circuits, i.e., the rotation angle can be any real number.
The encoding will be the same Boolean formula, but with a weight function dependent on the rotation angle~\cite{mei2024simulating}.
Thus, different rotation angles do not significantly affect the performance per depth of the approximate synthesis of the rotation gate $R_z$.
Hence, to demonstrate the use of approximate synthesis, we consider the quantum gate $R_z(\frac{\pi}{8})=\begin{smat} {e^{-i\pi/16} & 0 \\ 0 & e^{i\pi/16}}\end{smat}$. We choose to use the \combasis with the \emph{cyclic encoding} as it outperforms the \paulibasis with the \emph{cyclic encoding}, according to the results in \autoref{tab:syn_statistics} (Recall that \paulibasis with the \emph{linear-cyclic encoding} cannot perform \emph{approximate} synthesis).
We report statistics for each synthesis layer where an improvement in the achieved fidelity is observed.
The results are presented in \autoref{tab:synRZ}.
The corresponding output circuits are as follows.
\begin{itemize}
    \item \textbf{1 Layer:} with $\epsilon = 0.1$,  $\mathcal{C}' = (T)$
    \item \textbf{10 Layer:} with $\epsilon = 0.1$, $\mathcal{C}' = (T^\dagger, H, T^\dagger, H, T^\dagger, H, T^\dagger, H, T^\dagger, H)$
    \item \textbf{15 Layer:} with $\epsilon = 0.01$, $\mathcal{C}' = (H, T^\dagger, H, T, H, T, H, T^\dagger, H, T, H, T, H, T^\dagger, H)$
\end{itemize}

Other tools also target approximate synthesis.
For example, the optimal tool \texttt{mitms} and other non-optimal tools such as \texttt{gridsynth}\cite{ross2016ancillafree} use operator norm as a metric, while \cite{meister2023machinesynthesis} uses the \jfid{}, though the implementation of the latter one is not open-source.
We do not include a performance comparison because, to the best of our knowledge, no available tool performs optimal approximate synthesis with fidelity as a metric.

\section{Conclusion}
The results presented in this work demonstrate that maximum weighted model counting can be effectively employed for both exact and approximate quantum circuit synthesis. This approach benefits from the generality and extensibility of weighted model counting techniques as used in circuit simulation and equivalence checking, particularly their compatibility with diverse gate sets and representation bases. However, in its current form, the maximum weighted model counting exhibits limited performance in this new context.

There are, nonetheless, several promising directions for further development:
(1)~The current implementation builds upon a prototype version of \texttt{d4Max}, which lacks several optimizations to accommodate the extension to negative and complex weights. The design of dedicated algorithms tailored to this application remains an open avenue for future research.
(2)~Our layer-by-layer synthesis approach could benefit from incremental solving, as used in bounded model checking~\cite{biere2021bounded,GuntherW14}. Although incremental approaches have been studied for sampling and weighted model counting~\cite{meel2022inc,yang2024towards}, their application to maximum weighted model counting has not yet been explored.
(3)~Since the inception of the Model Counting Competition,\footnote{\url{https://mccompetition.org/}} continuous progress has been observed in the capabilities of model counters. By contributing our benchmarks to the community, we anticipate further methodological improvements that may enhance the reliability and performance of this approach.
(4)~The particular characteristics in the encoded CNF, such as an abundance of XOR clauses~\cite{mei2024gapp}, are not yet exploited by the solver. This could lead to significant performance gains~\cite{DBLP:conf/ijcai/KoricheLMT13}.
(5)~Lastly, incorporating symmetry considerations may further improve the efficiency of model counting, as has been demonstrated in related domains~\cite{DBLP:conf/ecai/BartKLM14}.

The depth-optimal synthesis presented here enables more efficient implementations of multi-qubit gates, a critical requirement in fault-tolerant and error-corrected quantum computing~\cite{jones2013logicsynthesisfaulttolerantquantum}.
Future work will also address minimizing the $T$ count, which remains an important objective in error-corrected architectures~\cite{kliuchnikov2016practical}.
While the current study focuses on the Clifford+T gate set, the underlying encoding is readily generalizable to other gate sets. Moreover, extending maximum weighted model counting to the stochastic SAT setting~\cite{littman2001stochastic,lee2017solving} could broaden its applicability to further quantum circuit optimization problems.
This line of work may also lead to the derivation of novel lower bounds for computationally hard problems in quantum circuit analysis, based on established results from the reasoning and satisfiability domains~\cite{Bannach2024On}.

\newpage
\bibliography{bib2doi} 

\begin{thebibliography}{10}

\bibitem{aaronson2008improved}
Scott Aaronson and Daniel Gottesman.
\newblock Improved simulation of stabilizer circuits.
\newblock {\em Physical Review A}, 70(5), nov 2004.
\newblock URL: \url{https://doi.org/10.1103%2Fphysreva.70.052328}, \href {https://doi.org/10.1103/physreva.70.052328} {\path{doi:10.1103/physreva.70.052328}}.

\bibitem{aliferis2006quantum}
Panos Aliferis, Daniel Gottesman, and John Preskill.
\newblock Quantum accuracy threshold for concatenated distance-3 codes.
\newblock {\em Quantum Inf. Comput.}, 6(2):97--165, mar 2006.
\newblock \href {https://doi.org/10.26421/QIC6.2-1} {\path{doi:10.26421/QIC6.2-1}}.

\bibitem{amy2013mitm}
Matthew Amy, Dmitri Maslov, Michele Mosca, and Martin Roetteler.
\newblock A meet-in-the-middle algorithm for fast synthesis of depth-optimal quantum circuits.
\newblock {\em {IEEE} Trans. Comput. Aided Des. Integr. Circuits Syst.}, 32(6):818--830, 2013.
\newblock \href {https://doi.org/10.1109/TCAD.2013.2244643} {\path{doi:10.1109/TCAD.2013.2244643}}.

\bibitem{arute2019quantum}
Frank Arute, Kunal Arya, Ryan Babbush, Dave Bacon, Joseph~C Bardin, Rami Barends, Rupak Biswas, Sergio Boixo, Fernando~GSL Brandao, David~A Buell, et~al.
\newblock Quantum supremacy using a programmable superconducting processor.
\newblock {\em Nature}, 574(7779):505--510, 2019.
\newblock URL: \url{https://www.nature.com/articles/s41586-019-1666-5}, \href {https://doi.org/10.1038/s41586-019-1666-5} {\path{doi:10.1038/s41586-019-1666-5}}.

\bibitem{audemard2022maxsat}
Gilles Audemard, Jean{-}Marie Lagniez, and Marie Miceli.
\newblock A new exact solver for (weighted) max{\#}sat.
\newblock In Kuldeep~S. Meel and Ofer Strichman, editors, {\em 25th International Conference on Theory and Applications of Satisfiability Testing, {SAT} 2022, August 2-5, 2022, Haifa, Israel}, volume 236 of {\em Leibniz International Proceedings in Informatics (LIPIcs)}, pages 28:1--28:20, Haifa, Israel, aug 2022. Schloss Dagstuhl - Leibniz-Zentrum f{\"{u}}r Informatik.
\newblock URL: \url{https://hal.science/hal-03765040}, \href {https://doi.org/10.4230/LIPIcs.SAT.2022.28} {\path{doi:10.4230/LIPIcs.SAT.2022.28}}.

\bibitem{Bannach2024On}
Max Bannach and Markus Hecher.
\newblock On weighted maximum model counting: complexity and fragments.
\newblock In {\em 36th {IEEE} International Conference on Tools with Artificial Intelligence, {ICTAI} 2024, Herndon, VA, USA, October 28-30, 2024}, pages 1--9. {IEEE}, 2024.
\newblock \href {https://doi.org/10.1109/ICTAI62512.2024.00010} {\path{doi:10.1109/ICTAI62512.2024.00010}}.

\bibitem{DBLP:conf/ecai/BartKLM14}
Anicet Bart, Fr{\'{e}}d{\'{e}}ric Koriche, Jean{-}Marie Lagniez, and Pierre Marquis.
\newblock Symmetry-driven decision diagrams for knowledge compilation.
\newblock In Torsten Schaub, Gerhard Friedrich, and Barry O'Sullivan, editors, {\em {ECAI} 2014 - 21st European Conference on Artificial Intelligence, 18-22 August 2014, Prague, Czech Republic - Including Prestigious Applications of Intelligent Systems {(PAIS} 2014)}, volume 263 of {\em Frontiers in Artificial Intelligence and Applications}, pages 51--56. {IOS} Press, 2014.
\newblock \href {https://doi.org/10.3233/978-1-61499-419-0-51} {\path{doi:10.3233/978-1-61499-419-0-51}}.

\bibitem{bernstein1993quantum}
Ethan Bernstein and Umesh~V. Vazirani.
\newblock Quantum complexity theory.
\newblock In S.~Rao Kosaraju, David~S. Johnson, and Alok Aggarwal, editors, {\em Proceedings of the Twenty-Fifth Annual {ACM} Symposium on Theory of Computing, May 16-18, 1993, San Diego, CA, {USA}}, pages 11--20. {ACM}, 1993.
\newblock \href {https://doi.org/10.1145/167088.167097} {\path{doi:10.1145/167088.167097}}.

\bibitem{biere2021bounded}
Armin Biere.
\newblock Bounded model checking.
\newblock In Armin Biere, Marijn Heule, Hans van Maaren, and Toby Walsh, editors, {\em Handbook of Satisfiability - Second Edition}, volume 336 of {\em Frontiers in Artificial Intelligence and Applications}, pages 739--764. {IOS} Press, 2021.
\newblock \href {https://doi.org/10.3233/FAIA201002} {\path{doi:10.3233/FAIA201002}}.

\bibitem{biere2009handbook}
Armin Biere, Marijn Heule, Hans van Maaren, and Toby Walsh, editors.
\newblock {\em Handbook of Satisfiability}, volume 185 of {\em Frontiers in Artificial Intelligence and Applications}.
\newblock {IOS} Press, 2009.

\bibitem{bravyi2019simulation}
Sergey Bravyi, Dan~E. Browne, Padraic Calpin, Earl~T. Campbell, David Gosset, and Mark Howard.
\newblock Simulation of quantum circuits by low-rank stabilizer decompositions.
\newblock {\em Quantum}, 3:181, sep 2019.
\newblock \href {https://doi.org/10.22331/q-2019-09-02-181} {\path{doi:10.22331/q-2019-09-02-181}}.

\bibitem{bravyi2016trading}
Sergey Bravyi, Graeme Smith, and John~A. Smolin.
\newblock Trading classical and quantum computational resources.
\newblock {\em Phys. Rev. X}, 6:021043, Jun 2016.
\newblock URL: \url{https://link.aps.org/doi/10.1103/PhysRevX.6.021043}, \href {https://doi.org/10.1103/PhysRevX.6.021043} {\path{doi:10.1103/PhysRevX.6.021043}}.

\bibitem{burgholzer2020advanced}
Lukas Burgholzer and Robert Wille.
\newblock Advanced equivalence checking for quantum circuits.
\newblock {\em {IEEE} Trans. Comput. Aided Des. Integr. Circuits Syst.}, 40(9):1810--1824, sep 2020.
\newblock URL: \url{http://dx.doi.org/10.1109/TCAD.2020.3032630}, \href {https://doi.org/10.1109/tcad.2020.3032630} {\path{doi:10.1109/tcad.2020.3032630}}.

\bibitem{calderbank1997qec}
A.~R. Calderbank, E.~M. Rains, P.~W. Shor, and N.~J.~A. Sloane.
\newblock Quantum error correction and orthogonal geometry.
\newblock {\em Phys. Rev. Lett.}, 78:405--408, Jan 1997.
\newblock URL: \url{https://link.aps.org/doi/10.1103/PhysRevLett.78.405}, \href {https://doi.org/10.1103/PhysRevLett.78.405} {\path{doi:10.1103/PhysRevLett.78.405}}.

\bibitem{chen2021quantumaf}
Senrui Chen, Sisi Zhou, Alireza Seif, and Liang Jiang.
\newblock Quantum advantages for pauli channel estimation.
\newblock {\em Physical Review A}, 2021.
\newblock URL: \url{https://api.semanticscholar.org/CorpusID:237213441}.

\bibitem{chen2023automata}
Yu{-}Fang Chen, Kai{-}Min Chung, Ondrej Leng{\'{a}}l, Jyun{-}Ao Lin, Wei{-}Lun Tsai, and Di{-}De Yen.
\newblock An automata-based framework for verification and bug hunting in quantum circuits.
\newblock {\em Proc. {ACM} Program. Lang.}, 7({PLDI}):1218--1243, jun 2023.
\newblock \href {https://doi.org/10.1145/3591270} {\path{doi:10.1145/3591270}}.

\bibitem{dawson2005solovaykitaev}
Christopher~M. Dawson and Michael~A. Nielsen.
\newblock The solovay-kitaev algorithm.
\newblock {\em Quantum Inf. Comput.}, 6(1):81--95, 2006.
\newblock \href {https://doi.org/10.26421/QIC6.1-6} {\path{doi:10.26421/QIC6.1-6}}.

\bibitem{nest2008classical}
M.~Van den Nest.
\newblock Classical simulation of quantum computation, the gottesman-knill theorem, and slightly beyond.
\newblock {\em arXiv:0811.0898}, 2008.
\newblock \href {https://doi.org/10.48550/arXiv.0811.0898} {\path{doi:10.48550/arXiv.0811.0898}}.

\bibitem{fowler2011steane}
Austin~G. Fowler.
\newblock Constructing arbitrary steane code single logical qubit fault-tolerant gates.
\newblock {\em Quantum Inf. Comput.}, 11(9{\&}10):867--873, sep 2011.
\newblock \href {https://doi.org/10.26421/QIC11.9-10-10} {\path{doi:10.26421/QIC11.9-10-10}}.

\bibitem{fowler2012surface}
Austin~G Fowler, Matteo Mariantoni, John~M Martinis, and Andrew~N Cleland.
\newblock Surface codes: Towards practical large-scale quantum computation.
\newblock {\em Physical Review A}, 86(3):032324, 2012.
\newblock URL: \url{https://link.aps.org/doi/10.1103/PhysRevA.86.032324}, \href {https://doi.org/10.1103/PhysRevA.86.032324} {\path{doi:10.1103/PhysRevA.86.032324}}.

\bibitem{gheorghiu2022tcount}
Vlad Gheorghiu, Michele Mosca, and Priyanka Mukhopadhyay.
\newblock {T}-count and {T}-depth of any multi-qubit unitary.
\newblock {\em npj Quantum Information}, 8(1), nov 2022.
\newblock URL: \url{http://dx.doi.org/10.1038/s41534-022-00651-y}, \href {https://doi.org/10.1038/s41534-022-00651-y} {\path{doi:10.1038/s41534-022-00651-y}}.

\bibitem{giles2013exact}
Brett Giles and Peter Selinger.
\newblock Exact synthesis of multiqubit {Clifford + T} circuits.
\newblock {\em Physical Review A}, 87(3):032332, 2013.
\newblock URL: \url{https://link.aps.org/doi/10.1103/PhysRevA.87.032332}, \href {https://doi.org/10.1103/PhysRevA.87.032332} {\path{doi:10.1103/PhysRevA.87.032332}}.

\bibitem{gosset2013algorithmtcount}
David Gosset, Vadym Kliuchnikov, Michele Mosca, and Vincent Russo.
\newblock An algorithm for the {T}-count, 2013.
\newblock URL: \url{https://arxiv.org/abs/1308.4134}, \href {https://arxiv.org/abs/1308.4134} {\path{arXiv:1308.4134}}, \href {https://doi.org/10.48550/arXiv.1308.4134} {\path{doi:10.48550/arXiv.1308.4134}}.

\bibitem{GuntherW14}
Henning G{\"{u}}nther and Georg Weissenbacher.
\newblock Incremental bounded software model checking.
\newblock In Neha Rungta and Oksana Tkachuk, editors, {\em 2014 International Symposium on Model Checking of Software, {SPIN} 2014, Proceedings, San Jose, CA, USA, July 21-23, 2014}, pages 40--47. {ACM}, 2014.
\newblock \href {https://doi.org/10.1145/2632362.2632374} {\path{doi:10.1145/2632362.2632374}}.

\bibitem{huang2019approximate}
Yifei Huang and Peter Love.
\newblock Approximate stabilizer rank and improved weak simulation of {C}lifford-dominated circuits for qudits.
\newblock {\em Phys. Rev. A}, 99:052307, May 2019.
\newblock URL: \url{https://link.aps.org/doi/10.1103/PhysRevA.99.052307}, \href {https://doi.org/10.1103/PhysRevA.99.052307} {\path{doi:10.1103/PhysRevA.99.052307}}.

\bibitem{ji2009non}
Zhengfeng Ji and Xiaodi Wu.
\newblock Non-identity check remains {QMA}-complete for short circuits.
\newblock {\em arXiv:0906.5416}, 2009.
\newblock URL: \url{https://arxiv.org/abs/0906.5416}, \href {https://doi.org/10.48550/arXiv.0906.5416} {\path{doi:10.48550/arXiv.0906.5416}}.

\bibitem{jones2013logicsynthesisfaulttolerantquantum}
N.~Cody Jones.
\newblock Logic synthesis for fault-tolerant quantum computers, 2013.
\newblock URL: \url{https://arxiv.org/abs/1310.7290}, \href {https://arxiv.org/abs/1310.7290} {\path{arXiv:1310.7290}}, \href {https://doi.org/10.48550/arXiv.1310.7290} {\path{doi:10.48550/arXiv.1310.7290}}.

\bibitem{kissinger2022simulating}
Aleks Kissinger and John van~de Wetering.
\newblock Simulating quantum circuits with zx-calculus reduced stabiliser decompositions.
\newblock {\em Quantum Science and Technology}, 7(4):044001, 2022.
\newblock \href {https://doi.org/10.48550/arXiv.2109.01076} {\path{doi:10.48550/arXiv.2109.01076}}.

\bibitem{kliuchnikov2013fast}
Vadym Kliuchnikov, Dmitri Maslov, and Michele Mosca.
\newblock Fast and efficient exact synthesis of single qubit unitaries generated by {C}lifford and {T} gates, 2013.
\newblock URL: \url{https://arxiv.org/abs/1206.5236}, \href {https://arxiv.org/abs/1206.5236} {\path{arXiv:1206.5236}}, \href {https://doi.org/10.48550/arXiv.1206.5236} {\path{doi:10.48550/arXiv.1206.5236}}.

\bibitem{kliuchnikov2016practical}
Vadym Kliuchnikov, Dmitri Maslov, and Michele Mosca.
\newblock Practical approximation of single-qubit unitaries by single-qubit quantum clifford and {T} circuits.
\newblock {\em IEEE Transactions on Computers}, 65(1):161--172, 2016.
\newblock \href {https://doi.org/10.1109/TC.2015.2409842} {\path{doi:10.1109/TC.2015.2409842}}.

\bibitem{kocia2020improved}
Lucas Kocia and Mohan Sarovar.
\newblock Improved simulation of quantum circuits by fewer {G}aussian eliminations.
\newblock {\em Physical Review A}, 103(2), 2021.
\newblock URL: \url{http://dx.doi.org/10.1103/PhysRevA.103.022603}, \href {https://doi.org/10.1103/physreva.103.022603} {\path{doi:10.1103/physreva.103.022603}}.

\bibitem{DBLP:conf/ijcai/KoricheLMT13}
Fr{\'{e}}d{\'{e}}ric Koriche, Jean{-}Marie Lagniez, Pierre Marquis, and Samuel Thomas.
\newblock Knowledge compilation for model counting: Affine decision trees.
\newblock In Francesca Rossi, editor, {\em {IJCAI} 2013, Proceedings of the 23rd International Joint Conference on Artificial Intelligence, Beijing, China, August 3-9, 2013}, pages 947--953. {IJCAI/AAAI}, 2013.
\newblock URL: \url{http://www.aaai.org/ocs/index.php/IJCAI/IJCAI13/paper/view/6574}.

\bibitem{kuncak2010complete}
Viktor Kuncak, Mika{\"{e}}l Mayer, Ruzica Piskac, and Philippe Suter.
\newblock Complete functional synthesis.
\newblock {\em ACM Sigplan Notices}, 45(6):316--329, 2010.
\newblock \href {https://doi.org/10.1145/1806596.1806632} {\path{doi:10.1145/1806596.1806632}}.

\bibitem{lee2017solving}
Nian-Ze Lee, Yen-Shi Wang, and Jie-Hong~R Jiang.
\newblock Solving stochastic boolean satisfiability under random-exist quantification.
\newblock In {\em IJCAI}, pages 688--694, 2017.
\newblock \href {https://doi.org/10.24963/ijcai.2017/96} {\path{doi:10.24963/ijcai.2017/96}}.

\bibitem{littman2001stochastic}
Michael~L Littman, Stephen~M Majercik, and Toniann Pitassi.
\newblock Stochastic boolean satisfiability.
\newblock {\em Journal of Automated Reasoning}, 27:251--296, 2001.
\newblock \href {https://doi.org/10.1023/A:1017584715408} {\path{doi:10.1023/A:1017584715408}}.

\bibitem{maslov2018bruhat}
Dmitri Maslov and Martin Roetteler.
\newblock Shorter stabilizer circuits via bruhat decomposition and quantum circuit transformations.
\newblock {\em IEEE Transactions on Information Theory}, 64(7):4729–4738, jul 2018.
\newblock URL: \url{http://dx.doi.org/10.1109/TIT.2018.2825602}, \href {https://doi.org/10.1109/tit.2018.2825602} {\path{doi:10.1109/tit.2018.2825602}}.

\bibitem{matteo2016parallel}
Olivia~Di Matteo and Michele Mosca.
\newblock Parallelizing quantum circuit synthesis.
\newblock {\em Quantum Science and Technology}, 1(1):015003, oct 2016.
\newblock URL: \url{https://dx.doi.org/10.1088/2058-9565/1/1/015003}, \href {https://doi.org/10.1088/2058-9565/1/1/015003} {\path{doi:10.1088/2058-9565/1/1/015003}}.

\bibitem{mei2024simulating}
Jingyi Mei, Marcello~M. Bonsangue, and Alfons Laarman.
\newblock Simulating quantum circuits by model counting.
\newblock In Arie Gurfinkel and Vijay Ganesh, editors, {\em Computer Aided Verification - 36th International Conference, {CAV} 2024, Montreal, QC, Canada, July 24-27, 2024, Proceedings, Part {III}}, volume 14683 of {\em Lecture Notes in Computer Science}, pages 555--578. Springer, 2024.
\newblock URL: \url{https://doi.org/10.1007/978-3-031-65633-0\_25}, \href {https://doi.org/10.1007/978-3-031-65633-0_25} {\path{doi:10.1007/978-3-031-65633-0_25}}.

\bibitem{mei2024eq}
Jingyi Mei, Tim Coopmans, Marcello~M. Bonsangue, and Alfons Laarman.
\newblock Equivalence checking of quantum circuits by model counting.
\newblock In Christoph Benzm{\"{u}}ller, Marijn J.~H. Heule, and Renate~A. Schmidt, editors, {\em Automated Reasoning - 12th International Joint Conference, {IJCAR} 2024, Nancy, France, July 3-6, 2024, Proceedings, Part {II}}, volume 14740 of {\em Lecture Notes in Computer Science}, pages 401--421. Springer, 2024.
\newblock URL: \url{https://doi.org/10.1007/978-3-031-63501-4\_21}, \href {https://doi.org/10.1007/978-3-031-63501-4_21} {\path{doi:10.1007/978-3-031-63501-4_21}}.

\bibitem{mei2024gapp}
Jingyi Mei, Jan Martens, and Alfons Laarman.
\newblock Disentangling the gap between quantum and {\#}sat.
\newblock In Chutiporn Anutariya and Marcello~M. Bonsangue, editors, {\em Theoretical Aspects of Computing - {ICTAC} 2024 - 21st International Colloquium, Bangkok, Thailand, November 25-29, 2024, Proceedings}, volume 15373 of {\em Lecture Notes in Computer Science}, pages 17--40, Berlin, Heidelberg, 2024. Springer.
\newblock URL: \url{https://doi.org/10.1007/978-3-031-77019-7\_2}, \href {https://doi.org/10.1007/978-3-031-77019-7_2} {\path{doi:10.1007/978-3-031-77019-7_2}}.

\bibitem{meister2023machinesynthesis}
Richard Meister, Cica Gustiani, and Simon~C Benjamin.
\newblock Exploring ab initio machine synthesis of quantum circuits.
\newblock {\em New Journal of Physics}, 25(7):073018, jul 2023.
\newblock URL: \url{https://dx.doi.org/10.1088/1367-2630/ace077}, \href {https://doi.org/10.1088/1367-2630/ace077} {\path{doi:10.1088/1367-2630/ace077}}.

\bibitem{nielsen2000quantum}
Michael~A Nielsen and Isaac~L Chuang.
\newblock Quantum information and quantum computation.
\newblock {\em Cambridge: Cambridge University Press}, 2(8):23, 2000.

\bibitem{niemann2020advanced}
Philipp Niemann, Robert Wille, and Rolf Drechsler.
\newblock Advanced exact synthesis of clifford+t circuits.
\newblock {\em Quantum Information Processing}, 19(9):317, Aug 2020.
\newblock \href {https://doi.org/10.1007/s11128-020-02816-0} {\path{doi:10.1007/s11128-020-02816-0}}.

\bibitem{paradis2024synthetiq}
Anouk Paradis, Jasper Dekoninck, Benjamin Bichsel, and Martin~T. Vechev.
\newblock Synthetiq: Fast and versatile quantum circuit synthesis.
\newblock {\em Proc. {ACM} Program. Lang.}, 8({OOPSLA1}):55--82, apr 2024.
\newblock \href {https://doi.org/10.1145/3649813} {\path{doi:10.1145/3649813}}.

\bibitem{raginsky2001fidelity}
Maxim Raginsky.
\newblock A fidelity measure for quantum channels.
\newblock {\em Physics Letters A}, 290(1-2):11--18, 2001.
\newblock URL: \url{https://www.sciencedirect.com/science/article/pii/S0375960101006405}, \href {https://doi.org/10.1016/S0375-9601(01)00640-5} {\path{doi:10.1016/S0375-9601(01)00640-5}}.

\bibitem{rengaswamy2018synthesis}
Narayanan Rengaswamy, A.~Robert Calderbank, Henry~D. Pfister, and Swanand Kadhe.
\newblock Synthesis of logical clifford operators via symplectic geometry.
\newblock In {\em 2018 {IEEE} International Symposium on Information Theory, {ISIT} 2018, Vail, CO, USA, June 17-22, 2018}, pages 791--795. {IEEE}, 2018.
\newblock \href {https://doi.org/10.1109/ISIT.2018.8437652} {\path{doi:10.1109/ISIT.2018.8437652}}.

\bibitem{ross2016ancillafree}
Neil~J. Ross and Peter Selinger.
\newblock Optimal ancilla-free clifford+t approximation of z-rotations.
\newblock {\em CoRR}, abs/1403.2975, 2016.
\newblock URL: \url{https://arxiv.org/abs/1403.2975}, \href {https://arxiv.org/abs/1403.2975} {\path{arXiv:1403.2975}}, \href {https://doi.org/10.48550/arXiv.1403.2975} {\path{doi:10.48550/arXiv.1403.2975}}.

\bibitem{schneider2023sat}
Sarah Schneider, Lukas Burgholzer, and Robert Wille.
\newblock A {SAT} encoding for optimal clifford circuit synthesis.
\newblock In Atsushi Takahashi, editor, {\em Proceedings of the 28th Asia and South Pacific Design Automation Conference, {ASPDAC} 2023, Tokyo, Japan, January 16-19, 2023}, ASPDAC '23, pages 190--195, New York, NY, USA, 2023. Association for Computing Machinery.
\newblock \href {https://doi.org/10.1145/3566097.3567929} {\path{doi:10.1145/3566097.3567929}}.

\bibitem{selinger2014efficient}
Peter Selinger.
\newblock Efficient {Clifford+T} approximation of single-qubit operators, 2014.
\newblock URL: \url{https://arxiv.org/abs/1212.6253}, \href {https://arxiv.org/abs/1212.6253} {\path{arXiv:1212.6253}}, \href {https://doi.org/10.48550/arXiv.1212.6253} {\path{doi:10.48550/arXiv.1212.6253}}.

\bibitem{sistla2023symbolic}
Meghana Sistla, Swarat Chaudhuri, and Thomas~W. Reps.
\newblock Symbolic quantum simulation with quasimodo.
\newblock In Constantin Enea and Akash Lal, editors, {\em Computer Aided Verification - 35th International Conference, {CAV} 2023, Paris, France, July 17-22, 2023, Proceedings, Part {III}}, volume 13966 of {\em Lecture Notes in Computer Science}, pages 213--225. Springer, Springer, 2023.
\newblock URL: \url{https://doi.org/10.1007/978-3-031-37709-9\_11}, \href {https://doi.org/10.1007/978-3-031-37709-9_11} {\path{doi:10.1007/978-3-031-37709-9_11}}.

\bibitem{thanos2023fast}
Dimitrios Thanos, Tim Coopmans, and Alfons Laarman.
\newblock Fast equivalence checking of quantum circuits of clifford gates.
\newblock In {\'{E}}tienne Andr{\'{e}} and Jun Sun, editors, {\em Automated Technology for Verification and Analysis - 21st International Symposium, {ATVA} 2023, Singapore, October 24-27, 2023, Proceedings, Part {II}}, volume 14216 of {\em Lecture Notes in Computer Science}, pages 199--216, Cham, 2023. Springer.
\newblock URL: \url{https://doi.org/10.1007/978-3-031-45332-8\_10}, \href {https://doi.org/10.1007/978-3-031-45332-8_10} {\path{doi:10.1007/978-3-031-45332-8_10}}.

\bibitem{Thanos2024automated}
Dimitrios Thanos, Alejandro Villoria, Sebastiaan Brand, Arend-Jan Quist, Jingyi Mei, Tim Coopmans, and Alfons Laarman.
\newblock Automated reasoning in quantum circuit compilation.
\newblock In {\em SPIN 2024}, 2024.
\newblock URL: \url{https://spin-web.github.io/SPIN2024/assets/preproceedings/SPIN2024-paper6.pdf}.

\bibitem{limdd}
Lieuwe Vinkhuijzen, Tim Coopmans, David Elkouss, Vedran Dunjko, and Alfons Laarman.
\newblock {LIMDD:} {A} decision diagram for simulation of quantum computing including stabilizer states.
\newblock {\em Quantum}, 7:1108, 2023.
\newblock \href {https://doi.org/10.22331/q-2023-09-11-1108} {\path{doi:10.22331/q-2023-09-11-1108}}.

\bibitem{meel2022inc}
Suwei Yang, Victor~C. Liang, and Kuldeep~S. Meel.
\newblock {INC:} {A} scalable incremental weighted sampler.
\newblock In Alberto Griggio and Neha Rungta, editors, {\em 22nd Formal Methods in Computer-Aided Design, {FMCAD} 2022, Trento, Italy, October 17-21, 2022}, volume~3, pages 205--213. TU Wien Academic Press, {IEEE}, 2022.
\newblock URL: \url{https://doi.org/10.34727/2022/isbn.978-3-85448-053-2\_27}, \href {https://doi.org/10.34727/2022/isbn.978-3-85448-053-2_27} {\path{doi:10.34727/2022/isbn.978-3-85448-053-2_27}}.

\bibitem{yang2024towards}
Suwei Yang and Kuldeep~S. Meel.
\newblock Towards projected and incremental pseudo-boolean model counting.
\newblock {\em CoRR}, abs/2412.14485, 2024.
\newblock \href {https://arxiv.org/abs/2412.14485} {\path{arXiv:2412.14485}}, \href {https://doi.org/10.48550/arXiv.2412.14485} {\path{doi:10.48550/arXiv.2412.14485}}.

\end{thebibliography}

\newpage
\appendix
\section{Appendix}\label{appendix}

\subsection{Proof of \autoref{prop:exact-eqcheck}}
\label{app:prof_exact}
\exacteqcheck*
\begin{proof}
First of all, from \cite[Cor.~1]{mei2024eq}, the circuit $\circuit$ is equivalent to the identity $I^{\otimes n}$ if and only if \autoref{eq:linear-encoding} holds.

Next, we show that all three equations \autoref{eq:linear-cyclic-encoding} and \autoref{eq:cyclic-encoding} in both bases are equivalent. 

For \autoref{eq:linear-cyclic-encoding}, as stated in Corollary 1 and the proof of Lemma 2 in \cite{mei2024eq}, we have 
\[
\wmc(F_{\mathcal{P}}(\vec{q}) \wedge F_\mathcal C\left(\vec{q}, \vec{q}' \right) \wedge F_{\mathcal{P}}(\vec{q}')) \leq 1.
\]
for all $\mathcal{P} \in \set{X_j, Z_j \mid j\in [n]}$.
Thus, we infer that
\[
\sum_{\mathcal{P}\in \set{X_j, Z_j \mid j\in [n]}}\wmc(F_{\mathcal{P}}(\vec{q}) \wedge F_\mathcal C\left(\vec{q}, \vec{q}' \right) \wedge F_{\mathcal{P}}(\vec{q}')) \leq 2n,
\]
where the value achieves $2n$ if and only if each of the summands achieves $1$.
Therefore \autoref{eq:linear-cyclic-encoding} is true if and only if \autoref{eq:linear-encoding} is true, as demonstrated in \cite[Prop. 1]{mei2024eq}.

For \autoref{eq:cyclic-encoding} in \paulibasis{}, the idea is similar.
The value of the weighted model count of \autoref{eq:cyclic-encoding} is equivalent to
\[
\sum_{\mathcal{P}\in\{X,Y,Z,I\}^{\otimes n}}\wmc(F_{\mathcal{P}}(\vec{q}) \wedge F_\mathcal C\left(\vec{q}, \vec{q} \right) \wedge F_{\mathcal{P}}(\vec{q})) \leq 4^n,
\]
which can achieve $4^n$ if and only if for all $4^n$ Pauli strings $\mathcal P$,
the weighted model counting of $\mathcal P$ achieves 1.
Since $\set{X_j, Z_j \mid j\in [n]}\subseteq \{X,Y,Z,I\}^{\otimes n}$,
we have $\autoref{eq:cyclic-encoding}\Rightarrow \autoref{eq:linear-encoding}$.
From \cite{thanos2023fast},
if two unitaries are equivalent over $\set{X_j, Z_j \mid j\in [n]}$,
they are equivalent over $\{X,Y,Z,I\}^{\otimes n}$.
We have \autoref{eq:linear-encoding}$\Rightarrow$ \autoref{eq:cyclic-encoding}. 
Therefore \autoref{eq:linear-encoding}$\Leftrightarrow$ \autoref{eq:cyclic-encoding}.

For \autoref{eq:cyclic-encoding} in \combasis{},
since
\[
\begin{aligned}
    U_\circuit = \lambda \cdot I _{2^n} 
    & \Leftrightarrow 
    \bra{b}U_\mathcal{\circuit}\ket{b} = \lambda \text{ for } b\in \{0,1\}^n
    \\
    & \Leftrightarrow 
    \sum_{b\in \{0,1\}^n}\bra{b}U_\mathcal{\circuit}\ket{b} = \lambda \cdot 2^n,
    \\
    & \Leftrightarrow 
    \sum_{b\in \{0,1\}^n}\Bigl(\#SAT_W\bigl(F_{\ket{b}}(\vec q) \wedge F_\circuit(\vec q, \vec q')\wedge F_{\ket{b}}(\vec q')\bigr)\Bigr) =\lambda \cdot 2^n,
    &&(\autoref{lem:coefficient})
    \\
    & \Leftrightarrow 
    \sum_{b\in \{0,1\}^n}\Bigl(\#SAT_W\bigl(F_{\ket{b}}(\vec q) \wedge F_\circuit(\vec q, \vec q')\wedge F_{I_{2^n}}(\vec q, \vec q')\bigr)\Bigr) =\lambda \cdot 2^n,
    \\
    & \Leftrightarrow \#SAT_W(F_\circuit(\vec q, \vec q')\wedge F_{I_{2^n}}(\vec q, \vec q')) =\lambda \cdot 2^n,
\end{aligned}
\]
where  $|\lambda|^2 = 1$,
we have \autoref{eq:cyclic-encoding} $\Leftrightarrow$ $\circuit\equiv I^{\otimes n}$.
\end{proof}

\subsection{Optimization rules}
\label{app:opt}

To optimize the synthesis encoding, we implement additional constraints.

The first set of rules ensures that we avoid redundant combinations of gates, such as $HH$, since they can be replaced with $I$ gates:

\begin{itemize}
    \item \textbf{Rule 1}: No two $H$ gates in a row on the same qubit, since $HH=I$:
    \[
    F_{R1}^d(P) = \bigwedge_{k\in[d-1]}\bigwedge_{i\in[n]}(\overline{p}_{H,i}^{k}\vee\overline{p}_{H,i}^{k+1})
    \]
    \item \textbf{Rule 2}: No $8$ $T$ gates in a row on the same qubit since $T^8=I$:
    \[
    F_{R2}^d(P) = \bigwedge_{k\in[d-7]}\bigwedge_{i\in[n]} \bigvee_{j\in[8]}\overline{p}_{T,i}^{k+j}
    \]
    \item \textbf{Rule 3}: No two $CX$ gates in a row on the same qubits since $CX_{i,j}\cdot CX_{i,j}=I$: 
    \[
    F_{R3}^d(P) = \bigwedge_{k\in[d-1]} \bigwedge_{i,j\in[n], j\neq i} (\overline{p}_{CX,i,j}^{k}\vee\overline{p}_{CX,i,j}^{k+1})
    \]
\end{itemize}

The second set of rules aims to have a canonical representation for a given set of gates. Our guideline is that every non-$I$ is pushed back as far as possible. This means not allowing any single qubit gate to follow an $I$ gate, other than $I$ itself. For two-qubit gates, we do not allow following $I$ on both qubits. 

\begin{itemize}
    \item \textbf{Rule 4}: No single qubit gates other than $I$ after an $I$ gate: 
    \[ 
    F_{R4}^d(P) = \bigwedge_{k\in[d-1]}\bigwedge_{i\in[n]}(p_{I,i}^{k} \Rightarrow (p_{I,i}^{k+1} \vee \bigvee_{j\in[n], j\neq i}(p_{CX,i,j}^{k+1} \vee p_{CX,j,i}^{k+1})))
    \]
    \item \textbf{Rule 5}: No $CX$ gate following $I$ gate on both qubits: 
    \[
    F_{R5}^d(P) = \bigwedge_{k\in[d-1]}\bigwedge_{i,j\in[n], j\neq i}(p_{CX,i,j}^{k+1} \Rightarrow (\overline{p}_{I,i}^{k} \vee \overline{p}_{I,j}^{k}))
    \]
\end{itemize} 
\end{document}